\documentclass{article}

\usepackage{amsmath,amssymb,amsthm}%
\usepackage{bm}%
\usepackage{stmaryrd}%
\usepackage{graphicx}%
\usepackage{hyperref}%
\usepackage{cite}%

\newcommand{\iid}{i.i.d.\ }%
\newcommand{\Fth}{F_{\mathrm{th}}}%
\DeclareMathOperator*{\argmax}{arg\,max}%

\newcommand{\cH}{\mathcal{H}}%
\newcommand{\cJ}{\mathcal{J}}%

\newcommand{\Z}{\mathbb{Z}}%
\newcommand{\R}{\mathbb{R}}%
\newcommand{\Q}{\mathbb{Q}}%

\newcommand{\bbE}{\mathbb{E}}%
\newcommand{\bbP}{\mathbb{P}}%

\renewcommand{\vec}[1]{\boldsymbol{ #1 }}

\theoremstyle{plain}%
\newtheorem{theorem}{Theorem}%
\newtheorem{lemma}[theorem]{Lemma}%
\newtheorem{corollary}[theorem]{Corollary}%
\newtheorem{proposition}[theorem]{Proposition}%
\theoremstyle{definition}%
\newtheorem{definition}[theorem]{Definition}%
\newtheorem{algorithm}[theorem]{Algorithm}%
\theoremstyle{remark}%
\newtheorem*{remark}{Remark}%

\numberwithin{equation}{section}%
\numberwithin{theorem}{section}%

\title{Exact results for a toy model exhibiting dynamic criticality}%

\author{%
  David C. Kaspar \\%
  Mathematics Department \\%
  University of California \\%
  Berkeley, CA 94720, USA \\%
  \href{mailto:kaspar@math.berkeley.edu}{kaspar@math.berkeley.edu}
  \and
  Muhittin Mungan \\%
  Physics Department \\%
  Bo\u{g}azi\c{c}i University \\%
  Bebek 34342 Istanbul, Turkey \\%
  \href{mailto:mmungan@boun.edu.tr}{mmungan@boun.edu.tr}}

\begin{document}

\maketitle

\begin{abstract}
  In this article we discuss an exactly solvable, one-dimensional,
  periodic toy charge density wave model introduced in [D.C.~Kaspar,
  M.~Mungan, EPL {\bf 103}, 46002 (2013)].  In particular, driving the
  system with a uniform force, we show that the depinning threshold
  configuration is an explicit function of the underlying disorder, as
  is the evolution from the negative threshold to the positive
  threshold, the latter admitting a description in terms of record
  sequences.  This evolution is described by an avalanche algorithm,
  which identifies a sequence of static configurations that are stable
  at successively stronger forcing, and is useful both for analysis
  and simulation.  We focus in particular on the behavior of the
  polarization $P$, which is related to the cumulative avalanche size,
  as a function of the threshold force minus the current force $(\Fth
  - F)$, as this has been the focus of several prior numerical and
  analytical studies of CDW systems.  The results presented are
  rigorous, with exceptions explicitly indicated, and show that the
  depinning transition in this model is indeed a dynamic critical
  phenomenon.
\end{abstract}

\section{Introduction}\label{sec:intro}

We consider an infinite chain of particles connected by springs, where
each particle is exposed to an external potential.  The potentials are
identical except for quenched random phase shifts.  Such systems
originally served as phenomenological models for charge density waves,
a quantum phenomenon observed in certain materials at low temperature
\cite{Gruner}, but are now considered model problems in the study of
the behavior of elastic manifolds in disordered media; {\em see}
\cite{Fisher98,Giamarchi,Nattermann} for reviews.

Under the influence of an external driving force, the particles move,
perhaps within a single well of the substrate, or from one well to
another.  If the external force is not too strong, the chain will,
after some change in shape, come to rest; in this situation we say
that the system is \emph{pinned}, as there are positions for the
particles on the substrate which prevent the force from advancing it
further.  If, on the other hand, the force is very strong, no
arrangement of the particles on the substrate is sufficient to arrest
its progress, and we have entered the \emph{sliding} regime.  The
transition from one regime to the other occurs at a critical value of
the driving force, known as the \emph{threshold force} $\Fth$.  The
behavior of the system near threshold, and in particular the
transition from static to dynamic states, has been a subject of
interest in diverse areas, such as flux line lattices in type II
superconductors \cite{Blatter}, fluid invasion in porous media
\cite{Wilkinson}, propagation of cracks \cite{Bouchaud, Alava}, as
well as models of friction and earthquakes \cite{Kawamura12}.

Fisher \cite{DSFisher83,DSFisher85} has argued that this depinning
transition is an example of a \emph{dynamic critical phenomenon}, a
phase transition with the external force as the control parameter.
There is evidence to support this claim:
\begin{itemize}
\item analysis \cite{snc90,snc91} of a different simplified model
  \cite{MCG87} for sliding particles with random friction, showing the
  divergence of \emph{strains} at the depinning threshold,
\item functional renormalization group calculations
  \cite{NarayanFisher92,LeDoussal02,Ertas}, and
\item extensive numerical simulation in dimensions $d = 1,2,3$.
  \cite{Littlewood86, Veermans90B, MyersSethna93, Rosso02, Jensen,
    Middleton93, NarayanMiddleton94}
\end{itemize}
show or strongly suggest that certain properties of the system near
threshold exhibit scaling behavior.  On the other hand, there are few
\emph{rigorous} results to rely upon.

In a short paper \cite{KMShort13}, the authors introduced a toy
version of a CDW model in one dimension which is exactly solvable: the
threshold state is an explicit function of the underlying disorder, as
is the externally forced evolution to threshold through intermediate
static configurations.  This permits a precise examination of certain
observables, particularly the cumulative {\em avalanche size}, which
is related to the CDW \emph{polarization}, and here we find the
tell-tale signs of a critical phenomenon.  In this article we provide
the proofs and further details of the results stated in
\cite{KMShort13}.

The paper is organized as follows.  In Section \ref{sec:prelim} we
describe the Fukuyama-Lee-Rice model for CDWs, and the \emph{toy
  model} approximation that results from truncating the range of
interactions.  We introduce also the observables we study as the
configurations in these systems are driven to threshold.  Next, in
Section \ref{sec:avalanche} we formalize the process of evolving a
given configuration to threshold, through a sequence of static
configurations, as the \emph{avalanche algorithm}.  A number of
associated results hold for both the toy model and the untruncated
version.  Section \ref{sec:toy} presents additional observations for
the toy model, which take particular advantage of the explicit
description of the threshold state available in this case.  Both
Sections \ref{sec:avalanche} and \ref{sec:toy} concern statements
which hold almost surely with respect to the underlying disorder; in
Section \ref{sec:stats} we turn to statistics.  Remarks regarding
numerics are found in Section \ref{sec:numerics}.  To develop the
preceding material free from distraction, we defer all proofs to
Section \ref{sec:proofs}.  Lastly, in Section \ref{sec:conclusion} we
discuss our work and its context in the existing literature, and
indicate remaining questions for future work.

\section{Preliminaries} \label{sec:prelim}

The Fukuyama-Lee-Rice \cite{FuLee,LeeRice} description of CDWs is
analogous to a bi-infinite chain of particles connected by springs,
where each particle is subject to a randomly shifted potential.  We
assume also the presence of an external force acting uniformly on all
the particles.  A formal Hamiltonian for such a system is
\begin{equation}\label{eqn:Hamiltonian}
  \cH(\{y_i\}) 
  = \sum_{i \in \Z} \frac{1}{2}(y_i - y_{i-1} - \mu)^2 
  + V(y_i - \alpha_i) - Fy_i.
\end{equation}
Each particle $i$ is constrained to move in only one direction; we
call its location along this line $y_i$.  We have assumed the springs
are Hookean with equilibrium length $\mu$, and normalized their common
stiffness.  The potential $V$ is 1-periodic, and each particle sees a
different random translate of it.  $F$ is the driving force applied
uniformly to all the particles.  Particular choices are suitable for
deriving exact formulas:
\begin{itemize}
\item Let $\alpha_i$ be \iid uniform $(-\frac{1}{2}, +\frac{1}{2})$.
  As $V$ is 1-periodic, we may as well regard our random shifts as
  elements of the circle, where the Lebesgue measure is a natural
  choice.
\item As in \cite{Aubry83,NarayanFisher92}, we select the potential
  $V$ as
  \begin{equation}\label{eqn:potential}
    V(x) = \frac{\lambda}{2}(x - \llbracket x \rrbracket)^2,
  \end{equation}
  $\llbracket x \rrbracket$ denoting the integer nearest to $x$.  The
  parameter $\lambda > 0$ reflects the relative strength of the
  potential $V$ and the springs.
\end{itemize}
Figure \ref{fig:par-cdw} illustrates the situation.

\begin{figure}[h!]
  \begin{center}
    \includegraphics[width=4in]{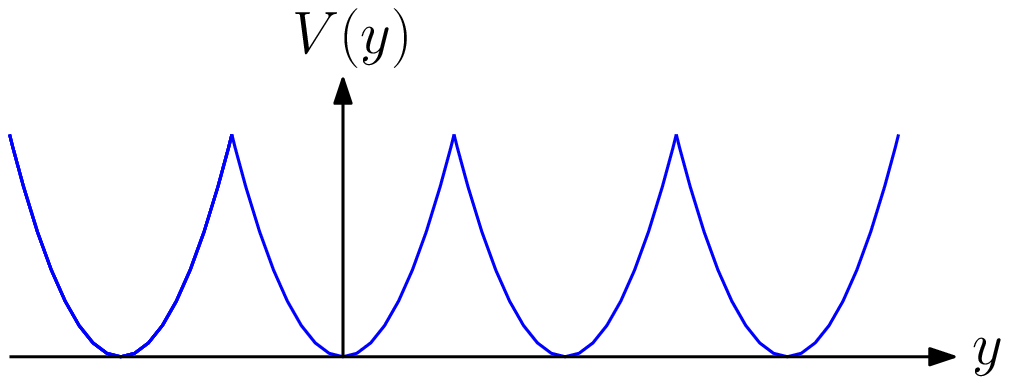} \\
    (a) \\
    \includegraphics[width=4in]{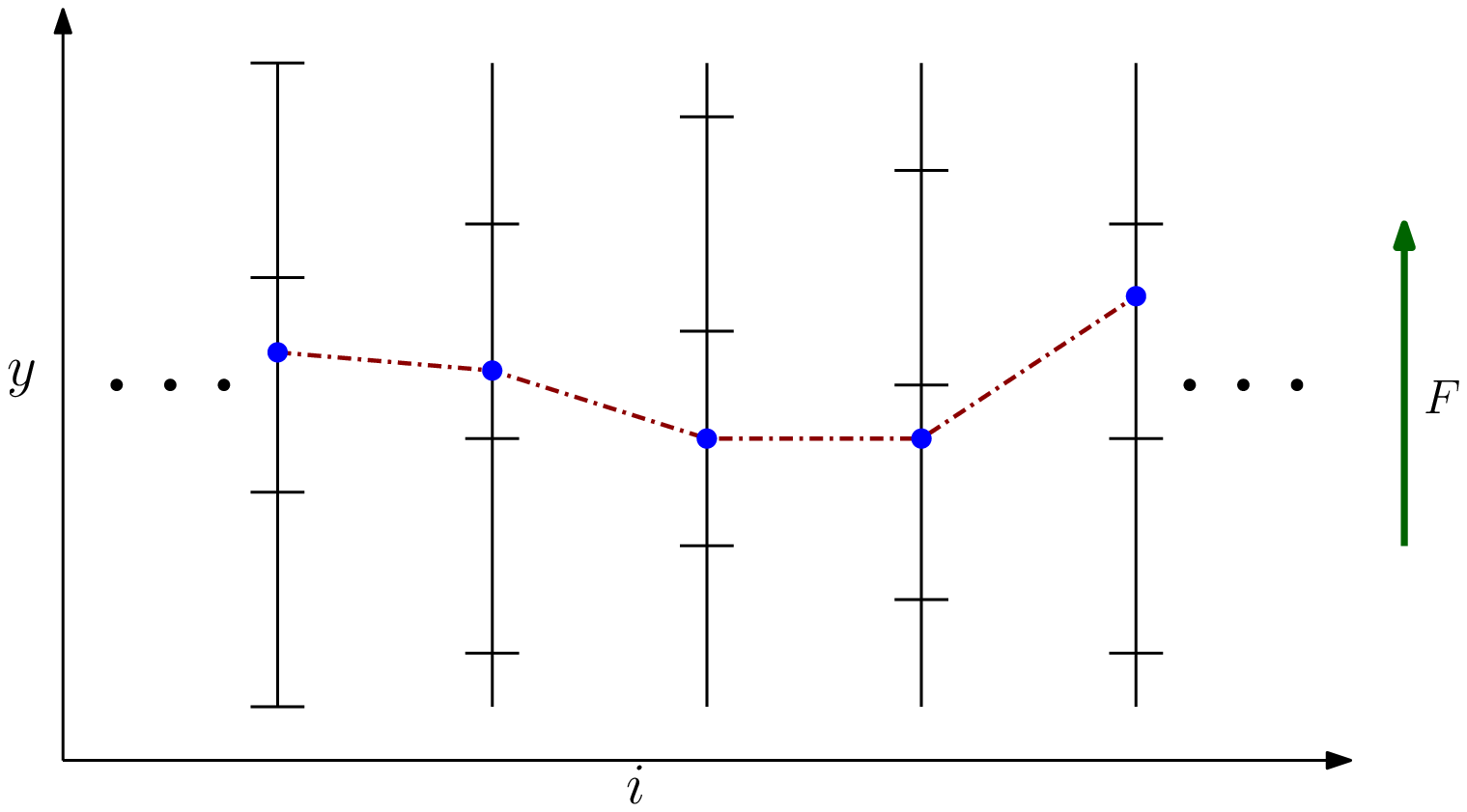} \\
    (b)
    \caption{(color online) Illustration (a) shows the shape of the
      potential $V(x)$, while (b) visualizes a portion of the
      bi-infinite chain of particles.  In (b) the particles are marked
      with blue dots and the ``springs'' connecting them are dashed
      red lines.  The vertical black lines show the sequence of
      potential wells seen by particle $i$, with horizontal markings
      to indicate the cusps of $V$.  An arrow indicates the direction
      of the external force $F$ exerted on the particles.}
    \label{fig:par-cdw}
  \end{center}
\end{figure}

It is possible to study the dynamics of such a system, under a
time-varying force, with a system of ODEs for inertialess particles
under relaxational dynamics: other authors such as \cite{LeeRice,
  Littlewood86, Middleton93, MyersSethna93, Veermans90B} have pursued
this approach.  Instead we will assume that the time scale at which
the external force is changing is much larger than that associated
with the relaxation of the particles, and therefore consider the
approach to threshold through intermediate static configurations.
Lemma \ref{lem:ZFAnoncrossing} and its analogue in the dynamic case,
the \emph{no passing} rule of \cite{Middleton93}, indicate some manner
of equivalence between these approaches.

Static configurations are those for which $\partial_{y_i} \cH = 0$ for
all $i$:
\begin{equation}\label{eqn:zeroforce1}
  -\Delta y_i + V'(y_i - \alpha_i) - F = 0.
\end{equation}
Here $\Delta$ denotes the discrete Laplace operator on sequences given
by $\Delta y_i = y_{i-1} - 2y_i + y_{i+1}$.  When using the potential
\eqref{eqn:potential} it is convenient to introduce the notation
\begin{subequations}
  \begin{align}
    m_i &\equiv \llbracket y_i - \alpha_i \rrbracket \in
    \Z \label{eqn:wellnumber} \\
    \tilde{y}_i &\equiv y_i - \alpha_i - m_i \in
    \left(-\textstyle\frac{1}{2},+\textstyle\frac{1}{2}\right]; \label{eqn:wellcoord}
  \end{align}
\end{subequations}
we refer to these as the \emph{well number} and \emph{well coordinate}
of $y_i$, the former indicating which parabolic well contains the
particle and the latter the displacement of the particle from the
center of its well.  Then \eqref{eqn:zeroforce1} can be re-expressed
as
\begin{equation}\label{eqn:zeroforce2}
  (\lambda - \Delta) \vec{y}  
  = \lambda(\vec{m} + \vec{\alpha} + F/\lambda).
\end{equation}
As in \cite{Aubry83}, we may treat $\vec{m}$ and $\vec{\alpha}$ as
given and solve this linear equation for $\vec{y}$.  It is important
to note, however, that the nonlinearity of this system has not
disappeared, but rather it becomes a consistency condition: after
computing $\vec{y}$ from $\vec{m}$, we must have $m_i = \llbracket y_i
- \alpha_i \rrbracket$ for all $i$.

Elementary techniques for linear recurrences applied to
\eqref{eqn:zeroforce2} give a formula for $\vec{y}$:
\begin{equation}\label{eqn:y}
  y_i = \frac{1 - \eta}{1 + \eta} \sum_{j\in\Z}
  \eta^{|i-j|}(m_j + \alpha_j) + \frac{F}{\lambda},
\end{equation}
where
\begin{equation}\label{eqn:eta}
  \eta = \frac{2}{2+\lambda+\sqrt{\lambda^2 + 4\lambda}} \in (0,1).
\end{equation}
This is the unique choice for which $y_i$ does not grow geometrically
as $|i| \to \infty$ even if $\vec{m} + \vec{\alpha}$ is bounded.
Noting that $\lambda \vec{\tilde{y}} = \Delta \vec{y} + F$ from
\eqref{eqn:zeroforce2}, it follows also that
\begin{equation}\label{eqn:tildey}
  \tilde{y}_i = \frac{\eta}{1-\eta^2} \sum_{j \in \Z}
  \eta^{|i-j|}(\Delta m_j + \Delta \alpha_j) + \frac{F}{\lambda}.
\end{equation}
Momentarily ignoring the relationship between $\vec{y}$ and $\vec{m}$,
observe that increasing or decreasing the driving force $F$ affects
the configuration by rigid translation.  This, and the explicit
formulas \eqref{eqn:y} and \eqref{eqn:tildey}, are the advantages of
the parabolic potential.

Taking a configuration (with $F = 0$, for example) and increasing $F$
causes the chain of particles to rigidly translate until for some $i$
we have $\tilde{y}_i = \frac{1}{2}$; any further increase causes this
particle to topple into the next well.  See Figure
\ref{fig:examplejump} for illustration.  For instance, if a jump
occurs at site $j$, the resulting change in well coordinates is
\begin{equation}\label{eqn:examplejump}
  \tilde{y}_i \to \tilde{y}_i - \delta_{ij} +
  \frac{1-\eta}{1+\eta} \eta^{|i-j|}
\end{equation}
provided that $\tilde{y}_i < \frac{1}{2}$ for all $i$ after the
change; otherwise, other particles will be pulled forward into their
next wells.  This process may terminate, resulting in a new static
configuration, or continue forever, in which case we understand the
configuration is no longer pinned and has entered the sliding regime.

\begin{figure}[h!]
  \begin{center}
    \includegraphics[width=4in]{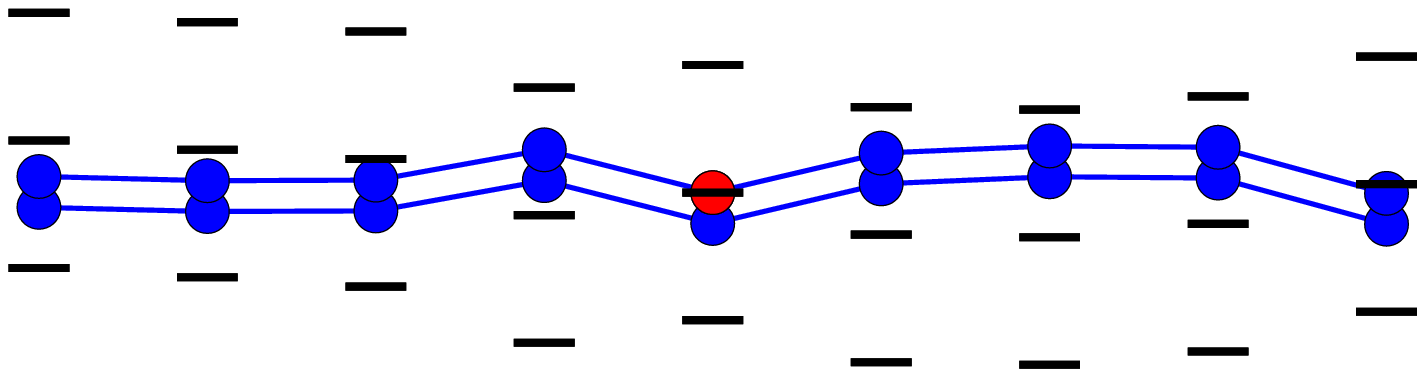} \\
    (a) \\
    \includegraphics[width=4in]{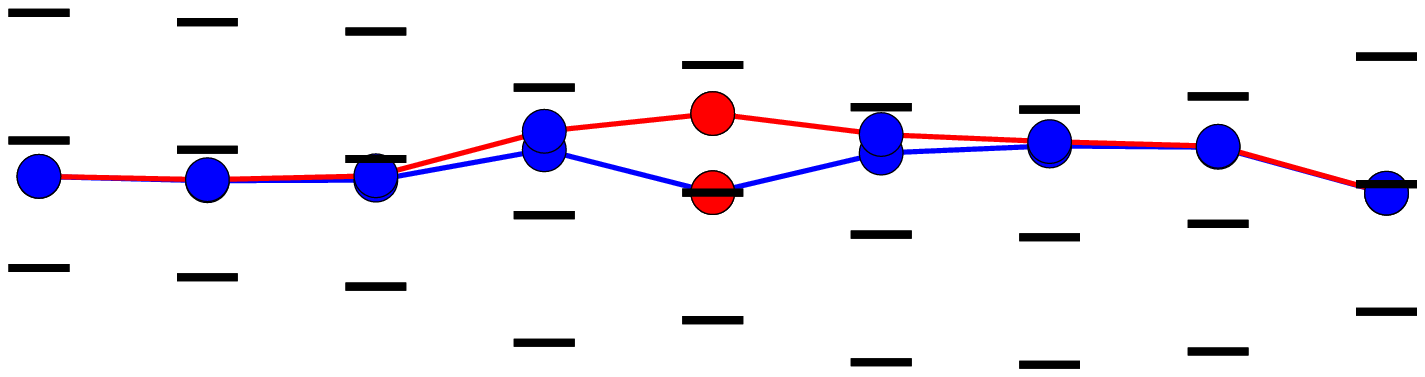} \\
    (b)
  \end{center}
  \caption{The configuration is (a) rigidly translated until a
    particle reaches the edge of its well (red) and then (b) this
    particle jumps into the next well, pulling all the other particles
    forward by an amount that decays geometrically moving away from
    the site that jumped.}
  \label{fig:examplejump}
\end{figure}

For this model we are interested in answering the following questions:
\begin{itemize}
\item At what $F$ does the system depin and enter the sliding regime?
  We call this $F$ the \emph{threshold force} and denote it $\Fth$.
\item What is the shape of the configuration just before it begins to
  slide?  As $|y_i - m_i| < 1$ for all $i$ by definition, the well
  numbers we observe just before the threshold, $\vec{m^+}$,
  sufficiently describe the large-scale shape.
\item How do various observables behave in terms of $\Fth - F$?  We
  are particularly interested in the \emph{polarization}, which is the
  spatial average (i.e.~average over $i$) of the change in $m_i$ as we
  evolve from some initial configuration to the first configuration we
  encounter that is stable at the current force $F$.
\end{itemize}
In subsequent sections we present theoretical and numerical results
for a \emph{finite} version of this system with periodic boundary
conditions.  For a system with $L$ particles, we take the well numbers
$m_i$ and the disorder $\alpha_i$ to be $L$-periodic sequences, the
latter still \iid within a single period.  Our most detailed results
are for an approximation we call the \emph{toy model}
\cite{KMShort13}.  For a strong potential, $\lambda$ is large and
$\eta$ is very small, and we have
\begin{equation}\label{eqn:toywell}
  \tilde{y}_i = \eta(\Delta m_i + \Delta \alpha_i) + F/\lambda
  + O(\eta^2).
\end{equation}
Dropping the $O(\eta^2)$ portion reduces the range of direct
interactions to nearest neighbors only.  In this case we can answer
very explicitly all the questions posed above.

\section{The avalanche algorithm}\label{sec:avalanche}

Our basic tool for both simulation and the derivation of rigorous
results is the \emph{avalanche algorithm}.  This takes as input a
static configuration, and produces a new static configuration which is
stable at higher force, if possible, in a manner intended to mimic the
result of increasing the force and finding the long time limiting
arrangement of the particles with an inertialess dynamics.  For the
$L$-periodic chain in both the toy model and the model with long-range
interaction, we describe this procedure in terms of the well numbers
$\vec{m}$ and well coordinates $\tilde{\vec{y}}$ from
\eqref{eqn:wellnumber} and \eqref{eqn:wellcoord}.

\begin{algorithm}[avalanche with force]\label{alg:avalanche}
  Given a valid configuration $\vec{m}$ in the environment specified
  by $\vec{\alpha}$ and $F$, we produce a new configuration
  $\vec{m}^*$ valid at a new $F^* \geq F$:
  \begin{itemize}
  \item[(A1)] Start with the current configuration: let $\vec{m}^* =
    \vec{m}$ (and correspondingly $\tilde{\vec{y}}^* =
    \tilde{\vec{y}}$) initially.
  \item[(A2)] Record the maximum well coordinate $\tilde{y}_{\max} =
    \max_i \tilde{y}_i$.
  \item[(A3)] Increase the force to $F^* = F + \lambda (\frac{1}{2} -
    \tilde{y}_{\max})$, and correspondingly adjust the well
    coordinates $\tilde{y}^*_i \to \tilde{y}^*_i + (\frac{1}{2} -
    \tilde{y}_{\max})$, bringing exactly one particle (in each period)
    to the cusp of the next well.
  \item[(A4)] Let $j = \argmax_i \tilde{y}^*_i$ and jump particle $j$
    (and its periodic equivalents) by incrementing $m^*_j \to m^*_j +
    1$ and suitably adjusting $\vec{\tilde{y}}^*$: for the full model,
    \begin{equation}\label{eqn:fullresponse}
      \tilde{y}^*_i \to \tilde{y}^*_i + 
      \begin{cases}
        \displaystyle\frac{-2\eta}{1+\eta} + \frac{1-\eta}{1+\eta}
        \frac{2\eta^L}{1-\eta^L} &\text{for } i = j \\ \\
        \displaystyle \frac{1-\eta}{1+\eta} \frac{\eta^{|i-j|} +
          \eta^{L-|i-j|}}{1-\eta^L} &\text{for } j < i < j + L
      \end{cases}
    \end{equation}
    and for the toy model,
    \begin{equation}\label{eqn:toyresponse}
      \tilde{y}^*_i \to \tilde{y}^*_i + 
      \begin{cases}
        +1\eta &\text{for } i = j\pm 1 \\
        -2\eta &\text{for } i = j,
      \end{cases}
    \end{equation}
    and extending these periodically.
  \item[(A5)] If $\tilde{y}^*_i > 1/2$ for any $i$, goto (A4).
  \end{itemize}
\end{algorithm}

\begin{remark}
  The formula \eqref{eqn:fullresponse} for updating $\tilde{\vec{y}}$
  has been adapted from \eqref{eqn:tildey} to respect the periodicity.
  For simulation purposes we might use \eqref{eqn:tildey} unaltered,
  summing only over nearest periodic representatives, at the cost of
  an $O(\eta^L)$ error.
\end{remark}

For this algorithm to be well-defined, we must verify that it does, in
fact, terminate.  The following result indicates that it does, and
gives the maximum number of jumps (A4) we might expect.  It also
establishes a useful property which will allow us to give a
\emph{centered} version of the algorithm, which is better numerically,
requiring fewer floating point operations, and better theoretically,
allowing us to recast the evolution in terms of a variational problem.

\begin{proposition}\label{prop:avalanche}
  The avalanche algorithm \ref{alg:avalanche} has the following
  properties:
  \begin{itemize}
  \item[(i)] It terminates after finitely many steps.
  \item[(ii)] All particles jump at most once: $\vec{m}^* \leq \vec{m}
    + \vec{1}$, the inequality holding componentwise.
  \item[(iii)] If $F \geq 0$ and $\eta < 1/3$, then for all $i$ the
    resulting configuration has
    \begin{equation}\label{eqn:bottomedge}
      \tilde{y}^*_i > - \frac{1}{2} + \frac{F^* - F}{\lambda}.
    \end{equation}
  \end{itemize}
\end{proposition}

Property (i) is immediate from (ii), which itself follows from a
consideration of \eqref{eqn:fullresponse} or \eqref{eqn:toyresponse}:
a particle $i$ which jumps once, decreasing $\tilde{y}_i$, will not
see sufficient increase in $\tilde{y}_i$ to exceed its original
height, even if all the other particles jump.  Property (iii) tells us
that the configuration $\vec{m}^*$ at force $F^*$ produced by the
algorithm remains a valid configuration---that is, has all its well
coordinates in $(-\frac{1}{2},+\frac{1}{2}]$---at the original force
$F$.  This illustrates that the models under consideration exhibit
both \emph{reversible and irreversible behavior}: increasing the force
from $0$ to some $F > 0$ may cause jumps, which are not undone if we
reduce the force back to $0$; on the other hand, the new configuration
we obtain reacts to values of the force in $[0,F]$ moving by rigid
translation only, i.e.\ reversibly.  It also allows us to write a
simpler algorithm which will produce the exact same\footnote{More
  precisely, the well numbers produced will be exactly the same, and
  the well coordinates will differ only by an overall translation
  applied uniformly to all particles.} sequence of configurations.

\begin{algorithm}[zero-force avalanche]\label{alg:ZFA}
  Given a configuration $\vec{m}$ in an environment specified by
  $\vec{\alpha}$ with $F = 0$, produce a new configuration $\vec{m}^*$
  valid at $F^* = 0$:
  \begin{itemize}
  \item[(ZFA1)] Let $\vec{m}^* = \vec{m}$.
  \item[(ZFA2)] Record $\tilde{y}_{\max} = \max_i \tilde{y}_i$.
  \item[(ZFA3)] Let $j = \argmax_i \tilde{y}^*_i$ and jump particle
    $j$ (and its periodic equivalents) by incrementing $m^*_j \to
    m^*_j + 1$ and correspondingly adjusting $\vec{\tilde{y}}^*$ as in
    \eqref{eqn:fullresponse} or \eqref{eqn:toyresponse}.
  \item[(ZFA4)] If $\tilde{y}^*_i > \tilde{y}_{\max}$ for any $i$,
    goto (ZFA3).
  \end{itemize}
  For brevity we refer to this algorithm as the \emph{ZFA}.  Note that
  the result has $\max_i \tilde{y}^*_i = \tilde{y}^*_{\max} \leq
  \tilde{y}_{\max}$.
\end{algorithm}

Middleton's \emph{no passing rule} \cite{Middleton93} is a
monotonicity property of the inertialess ODE system used to study CDWs
from a dynamic perspective: if $\vec{y}^1(t)$ and $\vec{y}^2(t)$ are
two solutions to $\dot{\vec{y}} = - \nabla \cH(\vec{y})$ where
$\vec{y}^1(t_0) \leq \vec{y}^2(t_0)$, then $\vec{y}^1(t) \leq
\vec{y}^2(t)$ for $t \geq t_0$.  Monotonicity results are an essential
tool for studying arrangements of chains of particles, even in a
purely static setting.  Consider, for example, the Aubry-Mather
treatment of the similar Frenkel-Kontorova model, explained very
nicely by Bangert \cite{Bangert88}.  That the ZFA has such a property
is necessary for our subsequent observations.

\begin{lemma}[ZFA noncrossing]\label{lem:ZFAnoncrossing}
  Let $\vec{m}^1 \leq \vec{m}^2$ be two configurations for either the
  full or toy model sharing the same environment $\vec{\alpha}$, and
  let $\vec{m}^{1*}$ and $\vec{m}^{2*}$ be the results of applying the
  ZFA to each of these.
  \begin{itemize}
  \item[(i)] If $\max_i \tilde{y}^1_i > \max_i \tilde{y}^2_i$, then
    $\vec{m}^{1*} \leq \vec{m}^2$.
  \item[(ii)] If $\max_i \tilde{y}^1_i = \max_i \tilde{y}^2_i$ and
    $m^1_j < m^2_j$ for $j = \argmax_i \tilde{y}^1_i$, then
    $\vec{m}^{1*} \leq \vec{m}^2$.
  \item[(iii)] If $\max_i \tilde{y}^1 \geq \max_i \tilde{y}^2_i$, then
    $\vec{m}^{1*} \leq \vec{m}^{2*}$.
  \end{itemize}
\end{lemma}

In each case above, the stated conditions give a bound on the well
coordinates of any particle $i$ for which $m^1_i = m^2_i$, which
prevents particle $i$ from jumping in cases (i) and (ii), or shows
that particle $i$ jumps for configuration 1 only if it jumps for
configuration 2.  The argument is very much the same as for the
dynamic version \cite{Middleton93}.

We now define the threshold states for the full and toy models.
Considering the above noncrossing result, the threshold configuration
should be that which minimizes $\max_i \tilde{y}_i$: another
configuration could not depin without first crossing this one.

\begin{definition}\label{def:threshold}
  In either the full model or the toy model, for a given environment
  $\vec{\alpha}$, a \emph{threshold configuration} is specified by
  well numbers $\vec{m}^+$ achieving
  \begin{equation}\label{eqn:optimization}
    \min_{\vec{m}} \max_i \tilde{y}_i
  \end{equation}
  where $\tilde{\vec{y}}$ is the vector of well coordinates
  corresponding to $\vec{m}$ at $F = 0$.  The \emph{threshold force}
  is
  \begin{equation}\label{eqn:threshforce}
    \Fth = \lambda \left(\frac{1}{2} - \min_{\vec{m}} \max_i
      \tilde{y}_i \right).
  \end{equation}
  Note that $\Fth$ is exactly the force required to bring one particle
  in the threshold configuration to the upper edge of its well.  Here
  and in the sequel, we compute well coordinates from well numbers at
  $F = 0$.
\end{definition}

\begin{remark}
  With standard Frenkel-Kontorova, one is interested in configurations
  which minimize energy, which consists (in the case of Hookean
  springs) of an $\ell^2$-difference of $\vec{y}$ and its translate by
  one, and the terms coming from the substrate potential.  Here, when
  considering a similar system in the presence of an increasing
  driving force, the relevant functional is of $\ell^\infty$-type.
\end{remark}

Our next result illustrates the utility of the ZFA as we try to
understand threshold behavior.

\begin{proposition}\label{prop:ZFAthresh}
  For both the full model and the toy model:
  \begin{itemize}
  \item[(i)] The threshold configuration $\vec{m}^+$ exists and is
    almost surely unique, up to translating all components of
    $\vec{m}^+$ by the same integer.
  \item[(ii)] Starting from $\vec{m} = \vec{0}$, the ZFA finds
    $\vec{m}^+$ in finitely many steps.
  \item[(iii)] The ZFA applied to $\vec{m}^+$ produces $\vec{m}^+ +
    \vec{1}$, and this property is unique to the family $\vec{m}^+ +
    \Z\vec{1}$.
  \end{itemize}
\end{proposition}

Existence of a minimizer in \eqref{eqn:optimization} is easy:
\eqref{eqn:zeroforce1} and periodicity can be used to bound $\max_i
|\Delta m_i|$, allowing us to exclude all but a finite set of
$\vec{m}$ (modulo uniform translation by integers).  Uniqueness is
also relatively routine, after using the noncrossing property of the
ZFA to reduce possible nonuniqueness to configurations which are
ordered and have well numbers differing by at most one.  Noncrossing
gives (ii), and the uniqueness, together with the fact that the ZFA
can never increase $\max_i \tilde{y}_i$, implies (iii).

We thus have a tool, the ZFA, for both the full and toy models, which
produces the threshold configuration that precedes the depinning
transition.  It achieves this by way of a sequence of physically
meaningful intermediate states, according to an algorithm which is
straightforward to implement and apply to generate numerical results.
In the next section we specialize to the toy model, where more can be
said.

\section{The toy model: explicit formulas}\label{sec:toy}

In the case of the toy model we find it convenient to introduce
\emph{rescaled well coordinates} $\vec{z}$ defined by
\begin{equation}\label{eqn:rescaled}
  \eta z_i = \tilde{y}_i.
\end{equation}
As in the previous section, we fix the external force $F = 0$.  In
this case, a jump at site $j$ as in step (iii) of Algorithm
\ref{alg:ZFA} results in
\begin{equation}\label{eqn:rescaledwell}
  m_j \to m_j + 1, \quad z_j \to z_j - 2, \quad z_{j\pm 1}
  \to z_{j\pm 1} + 1.
\end{equation}
Here we find a strong similarity between the toy model and sandpile
models (see \cite{Redig05} for an introduction), as already noted by
other authors working on similar CDW systems \cite{Tangetal87,
  MyersSethna93, NarayanMiddleton94}.  Indeed, for one-dimensional
sandpile models, the change to $\vec{z}$ in \eqref{eqn:rescaledwell}
is precisely the result of toppling at site $j$.  The existing
literature on sandpiles is extensive; see \cite{Turcotte99} for a
survey, and note that models with continuous heights have been
considered previously \cite{Zhang89}.  However, the authors are unable
to find an exact match for the toy model in prior work.  As noted in
\cite{KMShort13}, the toy model has periodic boundary, conserves the
sum of $\vec{z}$, evolves deterministically, changes by integers only,
and preserves the fractional part of the $\Delta \vec{\alpha}$.  We
discuss this connection further in Section \ref{sec:stats}.

For now, the similarity between the two is a sign to expect that the
toy model will permit exact results: the set of recurrent states of a
standard one-dimensional sandpile is rather trivial, and one might
hope that the toy model's persistent disorder does not introduce so
much complexity that things become intractable.  The primary result of
this section confirms this: the solution of the optimization problem
posed in Definition \ref{def:threshold} can be expressed explicitly.

\begin{theorem}\label{thm:toythresh}
  Let $S = \sum_{i=0}^{L-1} \llbracket \Delta \alpha_i \rrbracket$.
  The a.s.~unique threshold configuration for the toy model
  $\vec{m}^+$ satisfies
  \begin{equation}\label{eqn:toythresh}
    \Delta m^+_i = -\llbracket \Delta \alpha_i \rrbracket
    + J_i - \delta_{ik^+}
  \end{equation}
  where $\vec{J}$ is an integer vector selected as follows:
  \begin{itemize}
  \item \emph{Case $S \geq 0$.} $J_i = 1$ for the $S + 1$ positions
    $i$ which have smallest $\Delta \alpha_i - \llbracket \Delta
    \alpha_i \rrbracket$ and $J_i = 0$ otherwise.
  \item \emph{Case $S < 0$.}  $J_i = -1$ for the $|S| - 1$ positions
    $i$ which have largest $\Delta \alpha_i - \llbracket \Delta
    \alpha_i \rrbracket$ and $J_i = 0$ otherwise.
  \end{itemize}
  and $k^+$ is an index defined by
  \begin{equation}\label{eqn:divisibility}
    k^+ = \sum_{i=0}^{L-1} i(-\llbracket \Delta \alpha_i
    \rrbracket + J_i) \pmod{L}.
  \end{equation}
\end{theorem}
The proof is given in Section \ref{sec:proofs}, and, due to
Proposition \ref{prop:ZFAthresh}, amounts to checking that $\vec{m}^+$
is mapped to $\vec{m}^+ + \vec{1}$ by the ZFA.  To explore the
consequences of this explicit description, we first require some
notation.  Let
\begin{equation}\label{eqn:defects}
  \epsilon_i = \Delta m_i + \llbracket \Delta \alpha_i \rrbracket,
\end{equation}
and refer to those sites where $\epsilon_i \neq 0$ as \emph{defects}
with \emph{charge} $\epsilon_i$.  Write
\begin{equation}\label{eqn:omega}
  \omega_i = \Delta \alpha_i - \llbracket \Delta \alpha_i \rrbracket
\end{equation}
for the fractional part of $\Delta \alpha_i$, and let $\sigma$ be the
permutation of $\{0,1,\ldots,L-1\}$ which orders $\vec{\omega}$:
\begin{equation}\label{eqn:sigma}
  \omega_{\sigma(0)} < \omega_{\sigma(1)} < \cdots < \omega_{\sigma(L-1)}.
\end{equation}
Using this terminology, Theorem \ref{thm:toythresh} gives the
threshold force explicitly.

\begin{corollary}\label{cor:threshmax}
  For the toy model, the maximum $z^+_{\max}$ of the rescaled well
  coordinates (see \eqref{eqn:rescaled}) of the threshold
  configuration is
  \begin{equation}\label{eqn:threshmax}
    z^+_{\max} =
    \begin{cases}
      \omega_{\sigma(S)} + 1 &\text{if } S \geq 0 \text{ and
      } k^+ \neq \sigma(S) \\
      \omega_{\sigma(S-1)} + 1 &\text{if } S > 0 \text{
        and } k^+ = \sigma(S) \\
      \omega_{\sigma(L-1)} &\text{if } S = 0 \text{ and }
      k^+ = \sigma(0) \\
      \omega_{\sigma(L-|S|)} &\text{if } S < 0 \text{ and }
      k^+ \neq \sigma(L - |S|) \\
      \omega_{\sigma(L-|S|-1)} &\text{if } S < 0 \text{ and } k^+ =
      \sigma(L - |S|),
    \end{cases}
  \end{equation}
  and the corresponding threshold force $\Fth$ is
  \begin{equation}\label{eqn:toyFth}
    \Fth = \lambda\left(\frac{1}{2} - \eta z^+_{\max}\right).
  \end{equation}
\end{corollary}

\begin{remark}
  As we will see in Section \ref{sec:stats}, the cases $k^+ \in
  \{\sigma(S), \sigma(L-|S|)\}$ have probability tending to 0 as the
  system size $L \to \infty$.
\end{remark}

We wish to understand not only the threshold configuration but the
behavior of the system as we approach it.  The noncrossing property
Lemma \ref{lem:ZFAnoncrossing} of the ZFA implies that we may take
\emph{any} valid configuration and, by repeated application of this
algorithm, arrive at the threshold state.  During this process, we
track certain quantities associated with the system's evolution.

Of particular interest is the observable known as polarization, as
this has been the subject of several previous studies in CDW and
related models
\cite{NarayanFisher92,MyersSethna93,NarayanMiddleton94}.  Given an
initial state corresponding to some well numbers $\vec{m}^0$, applying
the ZFA produces a sequence of configurations (essentially)
terminating with $\vec{m}^+$.  Suppose that we record these
configurations, calling them $\vec{m}(\tau)$ for $\tau$ in some index
set $T$.  Then the \emph{polarization} is the function of $\tau$ given
by
\begin{equation}\label{eqn:polarization}
  P(\tau) 
  \equiv \frac{1}{L} \sum_{j=0}^{L-1} (m_j(\tau) - m^0_j).
\end{equation}
We write also $\Sigma(\tau) = L P(\tau)$, and call $\Sigma(\tau)$ the
\emph{cumulative avalanche size}.  In either case, the quantity under
consideration is the total number of particle jumps which have
occurred in the process of evolving from the initial state $\vec{m}^0$
to the current state $\vec{m}(\tau)$.

Among all possible initial conditions $\vec{m}^0$, two seem
particularly natural from a macroscopic perspective: we might begin
with \emph{flat} well numbers, $m^0_i = 0$ for all $i$, or we might
take the \emph{negative threshold} configuration, $\vec{m}^0 =
\vec{m}^-$, defined precisely below.  In the flat case we have only
statistics for the complete evolution without intermediate
configurations, which we discuss in Section \ref{sec:stats}.  For the
\emph{threshold-to-threshold} evolution, on the other hand, there is a
nice interpretation, in terms of record sequences, for each step of
the evolution, which we develop in the remainder of this section.

For both the toy model and the full model, given a realization
$\vec{\alpha}$, write $\vec{m}^+$ for a threshold configuration as
previously defined, and call it a \emph{$(+)$-threshold
  configuration}.  Define also a \emph{$(-)$-threshold configuration}
$\vec{m}^-$, which achieves
\begin{equation}\label{eqn:negthreshold}
  \max_{\vec{m}} \min_i \tilde{y}_i
\end{equation}
for $\tilde{y}_i$ the well coordinates corresponding to $\vec{m}$.
(Note that this can be obtained from the $(+)$-threshold configuration
with $\vec{\alpha}$ replaced with $-\vec{\alpha}$.)  We can adapt
\eqref{eqn:toythresh} to produce the negative threshold configuration
of the toy model, which maximizes $\min_i z_i$.  Define $J^-$ and
$k^-$ as follows:
\begin{itemize}
\item[(i)] \emph{Case $S > 0$.} $J^-_i = 1$ for the $S - 1$ positions
  $i$ which have smallest $\omega_i$, $J^-_i = 0$
  otherwise; 
\item[(ii)] \emph{Case $S \leq 0$.} $J^-_i = -1$ for the $|S| + 1$
  positions $i$ which have largest $\omega_i$, $J^-_i = 0$
  otherwise;
\end{itemize}
and $k^-$ is given in terms of $J^-$ by analogy with
\eqref{eqn:divisibility}.

In the toy model, several successive applications of the ZFA may have
initial jumps occurring at the same site.  This behavior will be
especially prevalent for the threshold-to-threshold evolution.  It
will be useful both intuitively and technically to view these
transitions in aggregate.  For a given \emph{non-threshold}
configuration $(\vec{m},\vec{z})$ with $i = \argmax_k z_k$, let $i_L$
and $i_R$ be the indices of sites closest to $i$ on the left and
right, respectively, for which
\begin{equation}\label{eqn:iLR}
  z_{i_L} + 1 \leq z_i \qquad \text{and} \qquad z_{i_R}
  + 1 \leq z_i.
\end{equation}
Define sets of indices $W_1,\ldots,W_\ell$ for\footnote{We use the
  notation $a\wedge b = \min(a,b)$ and $a \vee b = \max(a,b)$.
  Likewise, $x_+ = 0\vee x$.}  $\ell = (i - i_L) \wedge (i_R - i)$, by
\begin{equation}\label{eqn:waveindex}
  W_k = [i_L+k, i_R-k].
\end{equation}

\begin{proposition}\label{prop:avalanchewaves}
  The first $\ell$ iterations of the ZFA applied to $(\vec{m},
  \vec{z})$ as above cause jumps at sites with indices in
  $W_1,\ldots,W_\ell$, respectively.  We call this sequence an
  \emph{avalanche}, the individual iterations \emph{avalanche waves},
  and $i_L,i_R$ the \emph{left} and \emph{right extents} of the
  avalanche.  (The wave terminology has been borrowed from sandpile
  models \cite{Priezzhev94}.)  Throughout this process, site $i$
  remains the location of the maximum, and the original $z_i$ the
  maximum value, at least until after the $\ell^{\mathrm{th}}$
  iteration.  It follows that
  \begin{itemize}
  \item[(i)] The total number of jumps in the avalanche is $(i -
    i_L)(i_R - i)$.
  \item[(ii)] The resulting changes in the configuration $(\vec{m},
    \vec{z})$ are:
    \begin{subequations}
      \begin{align}
        z_{i_L} &\to z_{i_L} + 1 \\
        z_{i_R} &\to z_{i_R} + 1 \\
        z_i &\to z_i - 1 \\
        z_{i_L+i_R - i} &\to z_{i_L+i_R - i} - 1
      \end{align}
      and
      \begin{equation}
        \label{eqn:trapezoid}
        m_j \to m_j + (j - i_L)_+ - (j - i)_+ - (j - i - i_R + i_L)_+ +
        (j - i_R)_+.
      \end{equation}
    \end{subequations}
    In the case where $i - i_L = i_R -i $, the transition at $i$ is
    $z_i \to z_i - 2$.
  \end{itemize}
\end{proposition}

\begin{figure}[h!]
  \begin{center}
    \includegraphics[width=5in]{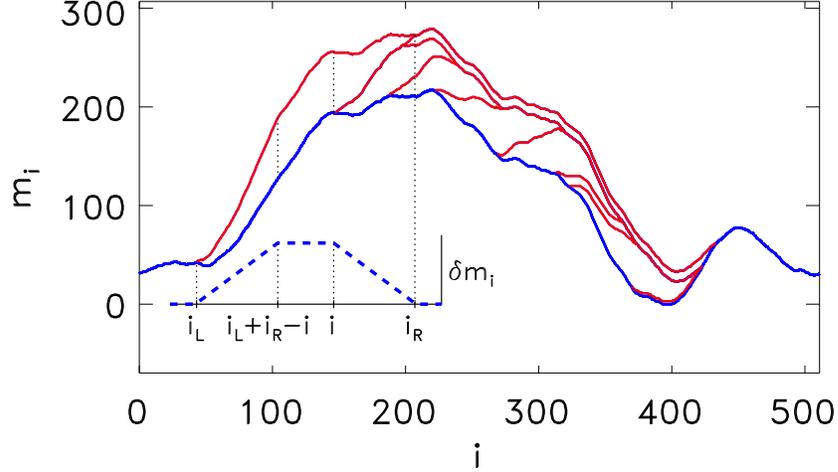} \\
    \vspace{2ex}
    \caption{Evolution of $\vec{m}$ under the ZFA starting from the
      negative threshold configuration (blue). The sequence of
      intermediate configurations reached by triggering of avalanches
      is shown in red. The topmost configuration is the positive
      threshold configuration. The inset displays the trapezoidal
      change resulting from one of the avalanches.}
    \label{fig:avalancheseries}
  \end{center}
\end{figure}

Note the change $\delta \vec{m}$ in $\vec{m}$ shown in
\eqref{eqn:trapezoid} is trapezoidal.  Figure
\ref{fig:avalancheseries} displays the changes resulting from several
avalanches.  To illustrate the threshold-to-threshold evolution for
the toy model, it is convenient to introduce
\begin{equation}
  \zeta_i = \omega_i + J^-_i
  \label{eqn:zetaomega}
\end{equation} 
and the permutation $\pi$ that orders $\zeta$:
\begin{equation}\label{eqn:zeta_order}
  \zeta_{\pi(0)} < \zeta_{\pi(1)} < \cdots < \zeta_{\pi(L-1)}.
\end{equation} 
Note $\zeta_{\pi(L-1)} - \zeta_{\pi(0)} < 1$. The $(\pm)$-threshold
configurations have
\begin{align}
  z^-_i &= \zeta_i + \delta_{ik^-} \label{eqn:zminus} \\
  z^+_i &= \zeta_i + \delta_{i \pi(0)} + \delta_{i \pi(1)} -
  \delta_{ik^+}, \label{eqn:zplus}
\end{align}
and, using the divisibility condition \eqref{eqn:divisibility},
$k^\pm$ are related by\footnote{Here and in the the following addition
  and subtraction of indices are mod $L$.}
\begin{equation}\label{eqn:kplus}
  k^+ = \pi(0) + \pi(1) - k^-.
\end{equation}

Observe that the ranks $\pi^{-1}$ of the $\vec{\zeta}$ suffice to
determine an avalanche's initial site and extents.  We represent a
given configuration $z_j$ by displaying the rank $\pi^{-1}(j)$ of
$\zeta_j$ and using over- or underlines to indicate additions by $\pm
1$ which are acquired as a result of jumps:
\begin{equation}
  \label{eqn:overlines}
  \begin{aligned}
    \overline{s} &\leftrightarrow z_{\pi(s)} = \zeta_{\pi(s)} + 1, \\
    \underline{s} &\leftrightarrow z_{\pi(s)} = \zeta_{\pi(s)} - 1.
  \end{aligned}
\end{equation}

As in \cite{KMShort13}, an example clarifies things.  Suppose that
$z^-$ has the rank representation
\begin{equation*}
  \begin{array}{cccccccccccc} 
    \ldots & 0 & 10 & 12 & 17 &
    \overline{15} & 16 & 18 & 11 & 13 & 1 & \ldots
  \end{array}
\end{equation*} 
so that $k^- = \pi(15)$.  The extents of the first avalanche are $k^-
- i_L = 2, i_R - k^- = 3$ and after the sites bracketed below have
jumped, the resulting configuration is
\begin{equation*}
  \begin{array}{cccccccccccc} 
    \ldots & 0 & 10 & \overline{12} & [\underline{17} & \overline{15} & 16 &
    \underline{18}] & \overline{11} & 13 & 1 & \ldots.
  \end{array}
\end{equation*} 
In the second wave, $k^-$ and $k^- + 1$ jump again, yielding
\begin{equation*}
  \begin{array}{cccccccccccc} 
    \ldots & 0 & 10 & \overline{12} & 17 & [15 & \underline{16}] & 18 &
    \overline{11} & 13 & 1 & \ldots,
  \end{array}
\end{equation*} 
and the avalanche is complete.  The remaining avalanches begin at the
sites ranked $12$, $11$, and $10$; the result is the positive
threshold configuration.

This example illustrates that the important sites in the
threshold-to-threshold evolution are the \emph{lower records}
\cite{Glick,RecordsBook}: given a sequence of values $X_1, X_2,
\ldots$, we say that $X_i$ is a \emph{lower record} if $X_i =
\min\{X_j : j \leq i\}$.  Using \eqref{eqn:iLR} and Proposition
\ref{prop:avalanchewaves} we see that avalanches are determined by the
locations of the lower records of the sequences
\begin{align}
  \mathcal{J}_L &= \zeta_{k^-}, \zeta_{k^- - 1}, \zeta_{k^- - 2},
  \ldots, \zeta_{\pi_L},
  \label{eqn:sequencesL} \\
  \text{and} \quad \mathcal{J}_R &= \zeta_{k^-}, \zeta_{k^- + 1},
  \zeta_{k^- + 2}, \ldots, \zeta_{\pi_R},
  \label{eqn:sequencesR}
\end{align}
where $\{\pi_L,\pi_R\} = \{ \pi(0), \pi(1) \}$ are the termination
sites. The evolution from negative to positive threshold terminates
when the avalanches reach $\pi_L$ and $\pi_R$.

We are most interested in the dependence of the polarization on $\Fth
- F$, the difference between the current force and that at
$(+)$-threshold.  For the zero-force description, the quantity that
serves this purpose is $X \equiv z_{\max} - z^+_{\max}$, the maximum
height of the current configuration minus that of the $(+)$-threshold
configuration.  We parametrize the configurations we see in the
threshold-to-threshold evolution by a nonnegative real quantity $x$:
$\vec{m}(x)$ is the first configuration we see for which $X \leq x$.
Note that this has the effect of skipping over the results of the
individual avalanche waves, because only complete avalanches give a
strict decrease in $X$.

By shifting indices, we can make $k^- = 0$; let $j_L(x)$ and $j_R(x)$
be the (noninclusive) left and right extents of the interval of sites
which have jumped in order to achieve $X \leq x$.  We select
\begin{equation}\label{eqn:jLRcenter}
  -L + j_R(x) < j_L(x) \leq 0 \leq j_R(x) < j_L(x) + L.
\end{equation}
Note that $j_L(x)$ and $j_R(x)$ are indices of the lower records from
the sequences \eqref{eqn:sequencesL} and \eqref{eqn:sequencesR},
respectively.  In the threshold-to-threshold evolution, $j_L(x)$ and
$j_R(x)$ are sufficient to characterize the shape of $\vec{m}(x) -
\vec{m}^0$, because this remains trapezoidal.  This follows because
\begin{itemize}
\item the result of any complete avalanche is a trapezoidal change,
  and
\item for the threshold-to-threshold evolution, starting with an
  overall trapezoidal change, the next avalanche is initiated at one
  of its convex corners, and terminates on one side at one of the
  concave corners.
\end{itemize}
Then the corresponding cumulative avalanche size $\Sigma(x)$ and
polarization $P(x)$ are
\begin{align}
  \Sigma(x) &= -j_L(x) j_R(x) \label{eqn:Sigmax} \\
  P(x) &= \frac{\Sigma(x)}{L}. \label{eqn:polarx}
\end{align}
To understand the threshold-to-threshold polarization as a function of
$\Fth - F$ amounts to understanding the statistics of the pair $j_L,
j_R$.  This and other probabilistic questions are addressed in the
next section.

\section{Statistical results}\label{sec:stats}

We begin by characterizing the variates $\omega_i = \Delta \alpha_i -
\llbracket \Delta \alpha_i \rrbracket$ introduced previously, as the
$(\pm)$-threshold configurations are explicit functions of these.  The
following proposition is not interesting itself, but gives some
indication how the choice we have made for the disorder enables the
subsequent results.

\begin{proposition}\label{prop:iidcond}
  The variates $\omega_i = \Delta \alpha_i - \llbracket \Delta
  \alpha_i \rrbracket$, $i = 0,\ldots,L-1$, have the joint
  distribution that results from taking \iid uniform
  $(-\frac{1}{2},+\frac{1}{2})$ variates and conditioning them to sum
  to an integer; by this we mean $\vec{\omega}$ is distributed
  according to the (normalized) surface measure on the intersection of
  the cube $(-\frac{1}{2},+\frac{1}{2})^L$ with the family of planes
  $x_0 + x_1 + \cdots + x_{L-1} \in \Z$.
\end{proposition}

Using the above it is easy to check that the one-dimensional marginals
are uniform $(-\frac{1}{2},+\frac{1}{2})$, and while
$\{\omega_i\}_{i=0}^{L-1}$ are dependent, removing just one of these
is enough to restore independence.  We apply the central limit theorem
for $L-1$ of these, and note that the variate omitted can alter the
sum by at most $\frac{1}{2}$.

\begin{corollary}\label{cor:cltforS}
  The sum $S = \sum_{i=0}^{L-1} \llbracket \Delta \alpha_i
  \rrbracket$, and hence the number of topological defects, behave as
  follows as $L \to \infty$.
  \begin{itemize}
  \item[(i)] As $L \to \infty$,
    \begin{equation}\label{eqn:cltforS}
      L^{-1/2} S = L^{-1/2} \sum_{i=0}^{L-1} \llbracket \Delta \alpha_i \rrbracket
    \end{equation}
    converges in distribution to a normal random variable with mean 0
    and variance $1/12$.
  \item[(ii)] The typical number of topological defects (sites where
    $\epsilon_i = \Delta m_i + \llbracket \Delta \alpha_i \rrbracket
    \neq 0$) scales like $L^{1/2}$.
  \end{itemize}
  We observe also \emph{numerically} that as $L \to \infty$, $\bbE \Fth
  \to \frac{\lambda}{2}(1 - \eta)$ and $L^{1/2}(\Fth - \bbE \Fth)$
  converges in distribution to a Gaussian with mean 0 and variance
  $(12L)^{-1}$
\end{corollary}

The rescaled well coordinates $\vec{z}^+$ at threshold are obtained by
the modification of $\vec{\omega}$ described in \eqref{eqn:toythresh}
and \eqref{eqn:divisibility}.  This modification does not preserve all
the properties of $\vec{\omega}$, but a particularly important one is
left intact.

\begin{theorem}\label{thm:exchangeable}
  The components $z^+_i$ of the vector $\vec{z}^+$ of centered,
  rescaled well-coordinates at threshold are exchangeable.
\end{theorem}

This leads quickly to a nice macroscopic description of the threshold
configurations as $L \rightarrow \infty$.  First, some physical
motivation: the \emph{strains}, which are the magnitudes of the forces
exerted by springs connecting the particles, are expected to diverge
at threshold \cite{snc90,snc91} in CDW systems.  This is possible
because we have assumed that the interaction between the particles can
survive any stress applied to it.  This is of course unphysical and
one expects that beyond a certain strain, plastic effects become
dominant. In the case of CDW systems, this plasticity gives rise to
phase slips: the springs yield once the strain reaches a certain
value.  If we intend to use the toy model to better understand such
behavior, we need to understand how the strains build up as a function
of the external force.  At present, we can at least characterize the
strains at threshold in a precise way.

Write $s_i = m_{i+1} - m_i$, $i = 0,\ldots,L-1$, for the strains in
the configuration indicated by $\vec{m}$.  Also let
\begin{equation}\label{eqn:strains}
  s^{(L)}(t) \equiv (12/L)^{1/2} s_{\lfloor Lt \rfloor} \quad (0 \leq t \leq 1)
\end{equation}
be the c\`{a}dl\`{a}g process obtained from $\vec{s}$ after central
limit rescaling.  A well known limit theorem for exchangeable variates
(found for instance in \cite{Kallenberg}) gives the distributional
limit of the processes $s^{(L)}$.

\begin{corollary}\label{cor:bridge}
  With $\vec{m} = \vec{m}^+$ and the corresponding threshold strains
  $\vec{s}$, as $L \to \infty$ the processes $s^{(L)}$ converge
  distributionally in the \emph{Skorokhod space} $\mathcal{D}([0,1])$
  (equipped with the $J_1$-topology) to a periodic Brownian motion
  with zero integral:
  \begin{equation}\label{eqn:bridge}
    B_0(t) \equiv B(t) - \int_0^1 B(r) \, dr \quad (0 \leq t \leq 1),
  \end{equation}
  where $B(t)$ is a standard Brownian bridge.  The process $B_0$ is
  Gaussian with zero mean, stationary under periodic translations of
  the interval $[0,1]$, with covariance given by
  \begin{equation}\label{eqn:bridgecov}
    \bbE B_0(0) B_0(t) = \frac{1}{12}(1 - 6t + 6t^2) \qquad (0 \leq t
    \leq 1).
  \end{equation}
\end{corollary}

Simulations of full CDW systems \cite{NarayanFisher92,
  NarayanMiddleton94} suggest that the total threshold-to-threshold
polarization scales like $P \sim L^{3/2}$.  The scaling limit of
Corollary \ref{cor:bridge} allows us to deduce this scaling for the
total polarization \emph{from flat initial condition to threshold.}
We compute
\begin{align*}
  P = \frac{1}{L} \sum_{i=0}^{L-1} m^+_i 
  &= \int_0^1 m^+_{\lfloor Lt \rfloor} \, dt
  = \int_0^1 \left\{\sum_{i=0}^{\lfloor Lt \rfloor} s_i - \min_{0\leq
      r \leq 1} 
    \sum_{i=0}^{\lfloor Lr \rfloor} s_i \right\} \, dt \\
  &= L \int_0^1 \left\{\int_0^t s_{\lfloor Lu \rfloor} \, du - \min_{0
      \leq r \leq 1} \int_0^r s_{\lfloor Lu \rfloor} \, du \right\} \,
  dt \\
  &= L^{3/2} \left\{ \int_0^1 \int_0^t L^{-1/2} s_{\lfloor Lu \rfloor}
    \, du \, dt - \min_{0 \leq r \leq 1} \int_0^r L^{-1/2} s_{\lfloor
      Lu \rfloor} \, du \right\}
  \\
  &= \frac{L^{3/2}}{\sqrt{12}} \left\{ \int_0^1 s^{(L)} (t) (1-t) \,
    dt - \min_{0 \leq r \leq 1} \int_0^r s^{(L)}(t) \, dt\right\}
\end{align*}
The functional on $\mathcal{D}([0,1])$ given by
\begin{equation}
  \psi(t) \mapsto \int_0^1 \psi(t) (1 - t) \, dt - \min_{0
    \leq r \leq 1} \int_0^r \psi(t) \, dt
\end{equation}
is continuous, so this yields a distributional limit for $L^{-3/2} P$.

The distributional limit for $\sqrt{12/L^3} P$ can be re-expressed in
terms of Brownian bridge:
\begin{equation}
  \label{eqn:bridgefunctional}
  \int_0^1 B_0(t) (1-t) \, dt - \min_{0 \leq r \leq 1} \int_0^r B_0(t)
  \, dt =
  \max_{0 \leq r \leq 1} \int_0^1 B(t) (\textstyle\frac{1}{2} - (t -
  r)
  - \mathbf{1}_{(0,r)}(t)) \, dt.
\end{equation}
Writing $\phi(t) = \frac{1}{2} - t$ for $0 \leq t \leq 1$, and
extending so that $\phi$ is $1$-periodic, the desired distribution is
that of
\begin{equation}
  \label{eqn:bridgepair}
  \max_{0 \leq r \leq 1} \int_0^1 B(t) \phi(t - r) \, dt = \max_{0
    \leq r \leq 1} \int_0^1 B(t + r) \phi(t) \, dt,
\end{equation}
extending $B(t)$ to be $1$-periodic.  Noting that $B(\cdot + r) -
B(r)$ has the same distribution as $B(\cdot)$, and that $\phi(t)$ is
orthogonal to constant functions, we find that
\begin{equation}
  \label{eqn:polarprocess}
  G(r) = \int_0^1 B(t) \phi(t - r) \, dt
\end{equation}
is a mean zero, stationary Gaussian process.  A straightforward
calculation gives
\begin{equation}
  \label{eqn:polarcov}
  \bbE G(0) G(r) = \frac{1}{720}(1 - 30 r^2 + 60 r^3 - 30 r^4) \qquad
  (0 \leq r \leq 1).
\end{equation}
In particular, $\bbE (G(r) - G(0))^2 \sim r^2$ as $r \to 0$, and a
result of Weber \cite{Weber89} applies to show there exists a constant
$c > 0$ so that
\begin{equation}
  \label{eqn:weber}
  c^{-1} t \sqrt{720} \Psi(t \sqrt{720})
  \leq \bbP\left\{\max_{0 \leq r \leq 1} G(r) > t \right\} \leq
  c t \sqrt{720} \Psi(t \sqrt{720})
\end{equation}
for all $t \geq 0$.  Here $\Psi(x)$ is the probability that a standard
normal random variable exceeds $x$.  It follows that the
distributional limit of $L^{-3/2} P$ has sub-Gaussian tail.  We are
unable to describe the distribution more precisely, and in general
distributions of maxima of Gaussian processes are known explicitly in
only a handful of cases \cite{Adler90}.  See Figure
\ref{fig:flatpolardist} for simulation results.

\begin{figure}[t!]
  \begin{center}
    \includegraphics[width=4in]{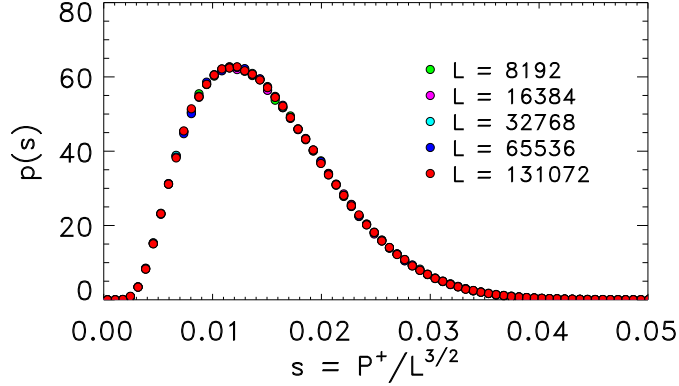}
    \caption{Simulated distribution for the flat-to-threshold
      polarization, rescaled by $L^{-3/2}$, for various $L$. The
      distributions were obtained from $10^6$ random realizations for
      each size.}
    \label{fig:flatpolardist}
  \end{center}
\end{figure}

For the threshold-to-threshold polarization in the toy model, our
description is considerably more detailed: instead of a single
quantity $P$, we have a function $P(x)$ defined in \eqref{eqn:polarx}
with a parameter $x$ indicating how close we are to the threshold ($x
= 0$).  Interestingly, the threshold-to-threshold polarization $P(0)
\sim L$, not $L^{3/2}$.

Relating the following proposition to the genuine $P(x)$ requires an
approximation which remains, at the moment, \emph{unjustified}, but
the result seems reasonable and matches very well our simulations.

\begin{proposition}\label{prop:polarization}
  Approximate $\cJ_L$ and $\cJ_R$ from \eqref{eqn:sequencesL} and
  \eqref{eqn:sequencesR} with \iid uniform
  $(-\frac{1}{2},+\frac{1}{2})$ variates sharing their first elements.
  Writing $x = u/L$, we obtain the finite-size scaling function
  $\Phi(u)$ for the cumulative avalanche size:
  \begin{equation}\label{eqn:polarscaling}
    \Phi(u) \equiv \lim_{L\to\infty} L^{-2}
    \bbE[\Sigma(u/L)]= \frac{6 - 4u + u^2 - 6e^{-u} - 2ue^{-u}}{u^4}.
  \end{equation}
  This is the result of averaging the distributional limit
  $\varsigma(u) \equiv \lim_{L \to \infty} \Sigma(u/L)/L^2$, which has
  density $p_u(s) = \bbP(\varsigma(u) \in ds)/ds$ given by
  \begin{equation}
    p_u(s) = 
    \int_{2\sqrt{s}}^1  dz \, e^{-zu} \frac{4+8u(1 - z)+ 2u^2(1-z)^2}{(z^2 - 4s)^{1/2}},
    \label{eqn:pdist}
  \end{equation}
  with support on the interval $[0,\frac{1}{4}]$.
\end{proposition}

Some remarks are needed to interpret this result.  First note that
there is no singularity in \eqref{eqn:polarscaling}: writing series
for the exponentials,
\begin{equation}\label{eqn:nopole}
  \Phi(u) = \frac{1}{12} - \frac{u}{30} + \frac{u^2}{120} -
  \frac{u^3}{105} + O(u^4)
\end{equation}
for $0 < u \ll 1$.  For $u \gg 1$ we have
\begin{equation}
  \Phi(u) \sim u^{-2}. 
  \label{eqn:phiuscaling}
\end{equation}
Noting the definition \eqref{eqn:polarscaling}, this shows that
$\Sigma$ (and not $P$, \cite{KMShort13}) exhibits finite size scaling
behavior: {\em i.e.} the graphs of $L^{-2} \bbE[\Sigma]$ {\em vs.} $ u
= XL$ for various $L$ asymptotically collapse to the graph of a the
scaling function $\Phi(u)$ and moreover, in the scaling regime $u \gg
1$, the dependence on $L$ drops out.  This is indeed confirmed by the
results of numerical simulations shown in Figure
\ref{fig:polarscaling}. The finite size scaling behavior implies that
the correlation length $\xi$ scales as
\begin{equation}
  \xi = X^{-1}.
\end{equation}
In terms of the underlying record process we can motivate this as
follows.  Given a current record $X$, the next record will occur on
average after $1/X$ sites.  Since all sites within this range are
forced to jump once the current record site initiates the next
avalanche, this defines the correlation length $\xi \sim X^{-\nu}$,
with exponent $\nu = 1$.  The crossover to the saturated regime occurs
when $\xi$ is comparable to $L$, namely $u = XL \sim L/\xi \sim 1$.
From \eqref{eqn:Sigmax}, the cumulative avalanche size is the product
of the left and right extents of sites which have jumped, and thus
scales as $X^{-2}$.  This exponent is traditionally denoted as
$-\gamma + 1$ \cite{DSFisher85,NarayanMiddleton94}, so that $\gamma =
3$.  The crossover behavior is clearly seen in Figure
\ref{fig:polarscaling}.

\begin{figure}[h!]
  \begin{center}
    \includegraphics[width=4in]{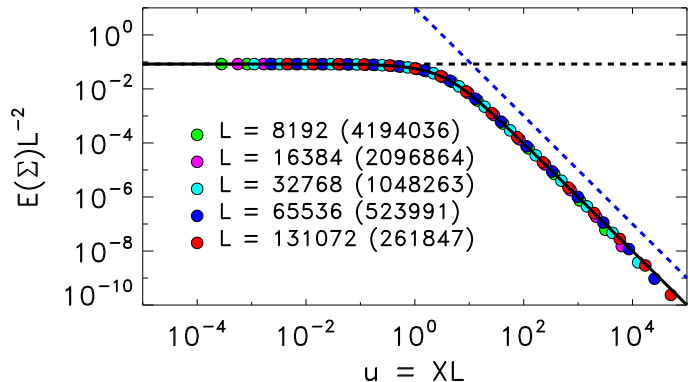}
    \caption{Numerical results for the expected cumulative jump size
      $\Sigma(X)$ vs.\ the reduced and rescaled force, $u = XL$, in
      the evolution from the negative to positive threshold
      configuration.  Symbol colors refer to different system sizes
      $L$, as indicated in the legend, with accompanying numbers of
      realizations in parentheses. The blue dashed line indicates a
      power-law with exponent $-2$, while the black line is the
      horizontal asymptote $\bbE({\Sigma})L^{-2} = 1/12$. The solid
      line is the theoretical finite-size scaling function, $\Phi(u)$,
      \eqref{eqn:polarscaling}.}
    \label{fig:polarscaling}
  \end{center}
\end{figure}

Two observations can be made about the distribution of $\varsigma(u)$.
We first identify a rescaling demonstrating its scale-free behavior
within the scaling regime, and then simplify \eqref{eqn:pdist} in the
case $u = 0$.  In each case, the results take familiar forms. Making a
change of variable
\begin{equation*}
  t = uz - 2u\sqrt{s},
\end{equation*}
the integral \eqref{eqn:pdist} becomes
\begin{equation}
  \label{eqn:pdistchanged}
  p_u(s) = 
  e^{-2u\sqrt{s}} \, \int_{0}^{u(1 - 2\sqrt{s})}  dt \, e^{-t} \,
  \frac{4+8(u - t - 2u\sqrt{s} ) + 2(u - t - 2u\sqrt{s})^2} { \sqrt{t(t + 4u\sqrt{s})} }.
\end{equation}
Scaling $\varsigma(u)$ such that
\begin{equation}
  \lim_{u \rightarrow \infty}  u^2\varsigma(u) \equiv  \mathfrak{a},
\end{equation}
the right endpoint of the interval of integration in
\eqref{eqn:pdistchanged} tends to $\infty$, and the numerator of the
fraction in the integrand is $2u^2$ to leading order.  By dominated
convergence as $u \to \infty$, the density $p(a) = \bbP(\mathfrak{a}
\in da)$ of the rescaled avalanche size $\mathfrak{a}$ is
\begin{equation}
  \label{eqn:paK0}
  p(a) = 
  2  \, e^{-2\sqrt{a}} \, \int_{0}^{\infty}  dt \, \frac{e^{-t}} { \sqrt{t(t + 4\sqrt{a})} }
  = 2 {\rm K_o}(2\sqrt{a}),
\end{equation}
where ${\rm K_o}$ is the modified Bessel Function, which decays at
large values of its argument as $e^{-2\sqrt{a}}/(2\sqrt{a})^{1/2}$.
For the $u$-values shown in Figure \ref{fig:sigmaDist}, the asymptotic
form \eqref{eqn:paK0} is indistinguishable from the exact result
\eqref{eqn:pdist}, explaining the collapse of the data.  The form of
the scaling variable $a$ can be understood by noting that $a = u^2 s =
X^2 \Sigma = \Sigma/\xi^2$; thus the avalanche sizes are measured in
units of $\xi^2$.

\begin{figure}[t!]
  \begin{center}
    \includegraphics[width=4.5in]{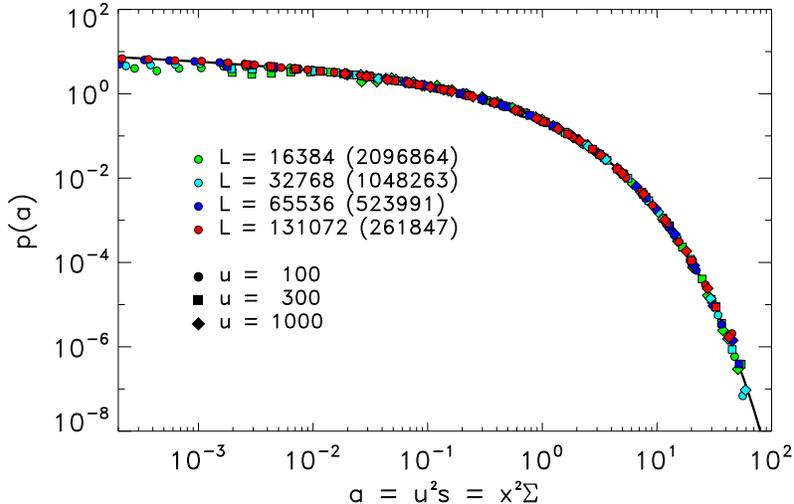}
    \caption{Numerical cumulative avalanche size distribution for
      various $L$ and $u$.  For large $u$, the distributions collapse
      when avalanche sizes are scaled as $a = u^2 s$. The solid line
      is \eqref{eqn:pdist}.  Symbol colors refer to different $L$, as
      indicated in the legend and the numbers of realizations are
      shown in parentheses. Symbol shapes refer to the different
      values of $u$ chosen.}
    \label{fig:sigmaDist}
  \end{center}
\end{figure}

Next, the density $p_0(s)$ for the complete threshold-to-threshold
polarization simplifies,
\begin{equation}
  p_0(s) = 2 \ln \frac{1 + \sqrt{1-4s}}{1 - \sqrt{1-4s}}.
  \label{eqn:PSigmath2th} 
\end{equation}
The distribution \eqref{eqn:PSigmath2th} matches exactly the avalanche
size distribution of Dhar's Abelian sandpile model in 1d
\cite{Dhar,RS92,Pruessner}. This is not a coincidence, as we now
explain.

In the sandpile model, the height parameter $h_i$ can take only
nonnegative integer values, and is stable only if $h_i \in \{0,1\}$
for all $i$.  Given a 1d sandpile of length $\mathcal{L}$ with sites
labeled $1$ to $\mathcal{L}$ and some stable initial configuration, a
site $i$ is selected at random and a grain of sand is added so that
$h_i \rightarrow h_i + 1$. The toppling rules of the model are as
follows:
\begin{itemize}
\item[(i)] Find any $j$ such that $h_j \ge 2$, set $h_j \to h_j - 2$,
  and $h_{j\pm 1} \to h_{j\pm 1} + 1$.
\item[(ii)] If $h_i \ge 2$ for some $i$, goto (i).
\end{itemize}
It is useful to add {\em pockets}, sites $i = 0$ and $i= \mathcal{L} +
1$, which we do not consider as part of the sandpile. One of the
grains which topples from site $1$ or $\mathcal{L}$ will fall into a
pocket and is lost.  We fix $h_0 = h_{\mathcal{L} + 1} = 0$.

The set of \emph{recurrent} states $\mathcal{R}$ consists of the
$\mathcal{L}+1$ configurations which have $h_i = 1$ for all but at
most one site (where it is $0$).  It can be shown \cite{Dhar,
  RS92,Redig05} that:
\begin{itemize}
\item [(a)] Starting with any configuration in $\mathcal{R}$ and
  adding $1$ at any site, the sandpile algorithm produces a result in
  $\mathcal{R}$.
\item [(b)] Under dynamics which consist of adding $1$ at a site
  chosen uniformly at random and stabilizing, the uniform distribution
  on $\mathcal{R}$ is invariant.
\item [(c)] Starting with a recurrent state and adding at site $k$
  with $h_k = 1$, the {\em active region} where topplings occur is the
  interval containing $k$ bounded by the closest sites $i$, possibly
  pockets, to the left and right of $k$ at which $h_i = 0$.  As the
  boundaries $\sigma_L, \sigma_R$ are excluded from the interval for
  the toy model, so are the boundaries where $h_i = 0$ excluded from
  the active region.
\end{itemize}
As we have seen in the previous sections, the evolution under the ZFA
is a dynamics of moving over- or underlines in the rank diagram.  In
the negative threshold configuration of the toy model we encounter
three types of sites, those with an overbar (corresponding to $h = 2$
sites), those without a bar ($h =1$ sites), and those with an underbar
($h = 0$ sites). After aligning the active regions of both models,
which have identical size if we set $\mathcal{L} = L - 2$, we obtain a
correspondence between a negative threshold configuration of the toy
model and a recurrent state of the sandpile.

We observe also the equivalence of the total threshold-to-threshold
evolution of the toy model and the stabilization of a recurrent
sandpile configuration when a single grain is added.  The key point is
that the {\em totality} of the iterated ZFA evolution is Abelian
\cite{KMShort13}.  If we set $z^+_{\rm max} = \max_j z^+_j$, the
maximum height of the positive threshold configuration, and then jump
all sites which have $z_j > z^+_{\rm max}$, then:
\begin{itemize}
\item all particles in the active region will be forced to jump at
  least once,
\item the order in which those sites with $z_j > z^+_{\max}$ are
  jumped is immaterial if we are concerned only with the final result,
  and
\item the positive threshold configuration is ultimately reached.
\end{itemize}
This is equivalent to running the BTW/Dhar sandpile algorithm on the
sites with overbars, which preserves the correspondence between
sandpile and toy model configurations. However, this map discards the
ordering and values of the well coordinates $z_i$, which in turn drive
the evolution towards threshold in the ZFA and thereby give rise to a
family of distributions \eqref{eqn:pdist}.

\section{Numerics} \label{sec:numerics}

The toy model ZFA has a very fast numerical implementation which we
now describe. For the threshold-to-threshold evolution, the negative
threshold configurations are generated following
\eqref{eqn:negthreshold}, and the random permutation $\pi$ from
\eqref{eqn:zeta_order} is obtained.  The evolution proceeds in units
of avalanches using Proposition \ref{prop:avalanchewaves}(iii) and the
rank representation of configurations, as outlined in the discussion
following the proposition.  We therefore only have to keep track of
the locations of the over- and underlines which involves simple
integer arithmetic.  This implementation is fast, since instead of
individual jumps we deal with avalanches and the expected number of
avalanches occurring during threshold-to-threshold evolution turns out
to scale as $\ln L$, which is what one expects, since a record
breaking process underlies the evolution from negative to positive
threshold. An explicit formula for the distribution of the number of
steps can be derived \cite{Terzi13}.

At the end of each avalanche we record various statistics, such as the
maximum of $z_i$, the cumulative number of jumps that have occurred at
a given site, and the size of the current avalanche.  All numerical
results presented here were obtained without parallelization on single
processors of an HP Z800 workstation. The longest run of about 262000
realizations of a size $L = 131072$ system took 4 hours.

The control parameter for the approach to threshold is the difference
between the sample-dependent threshold force $F_{\rm th}$ and the
current force $F$. For the ZFA, which holds the force fixed at 0, the
appropriate parameter is
\begin{equation}
  X = \max_i z_i - \max_i z^{+}_i = \max_i z_i - \zeta_{\pi(1)}, 
\end{equation}
where the last equality follows from \eqref{eqn:zplus}.

The values of $X$ are recorded at the end of each avalanche. In the
course of threshold-to-threshold evolution, we obtain a decreasing
sequence $X_\tau$ of $X$ values, where $\tau$ indexes the
avalanches. Following the definition of the corresponding processes,
\eqref{eqn:jprocdef}, if we want to obtain statistics for a particular
value $x$, the contributing avalanches $\tau$ will be those which
satisfy $X_\tau \leq x < X_{\tau-1}$, since the corresponding
configuration driven under an external force could have been
translated by this amount $x$ without incurring any particle jumps.
This is how the $x$-dependent avalanche size distributions and their
expectation values have been obtained in Figures
\ref{fig:polarscaling} and \ref{fig:sigmaDist}.

\begin{figure}[t!]
  \begin{center}
    \includegraphics[width=4.5in]{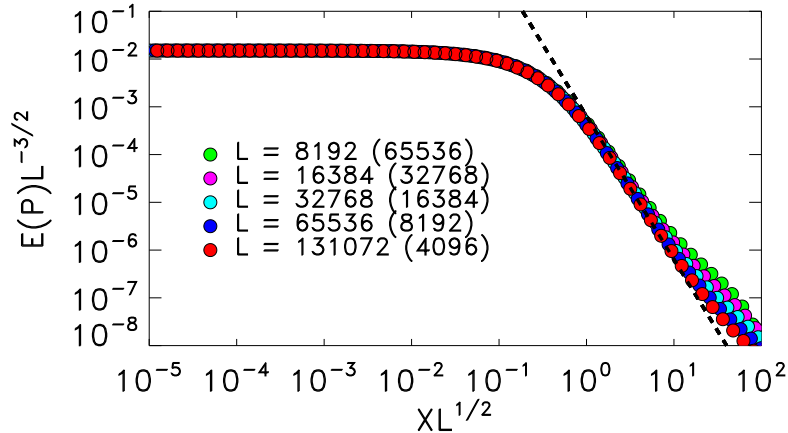}
    \caption{Numerical results for the expected cumulative jump size
      $P(X)$ vs.\ the reduced and rescaled force $u = XL$ in the
      evolution to threshold, starting from flat initial conditions,
      $\vec{m} = \vec{0}$. Colors refer to different system sizes $L$,
      as indicated in the legend, with accompanying numbers of
      realizations in parentheses. The dashed line is a power law with
      exponent $-3$.}
    \label{fig:polarscaling_flat2th}
  \end{center}
\end{figure}

We have also simulated the evolution from a flat initial
configuration, $\vec{m} = \vec{0}$, to positive threshold. The
evolution proceeds again by avalanches and Figure
\ref{fig:polarscaling_flat2th} shows our numerical results. The curves
for different system sizes collapse for values of $u = XL <10$ under
the scaling of the axes as indicated in the figure.  The scaling of
$P$ with $L^{-3/2}$ at $u = 0$ is in agreement with the prediction
following Corollary \ref{cor:bridge}. The scaling of the abscissa as
$XL^{1/2}$ suggests that the correlation lengths $\xi$ scales now as
$\xi \sim X^{-2}$. We return to a discussion of this result in the
conclusion.

\section{Proofs}\label{sec:proofs}

In this section we provide proofs for the results stated in the
preceding text, in order of appearance.

\begin{proof}[Proposition \ref{prop:avalanche}]
  That (i) the number of jumps is finite, and in fact bounded by $L$,
  is immediate from (ii) $\vec{m}^* \leq \vec{m} + 1$, so we proceed
  to the latter.  We argue inductively: suppose that after some
  execution of (A4) we have well numbers $\vec{m}'$ and well
  coordinates $\vec{\tilde{y}}'$, and that $\vec{m}' \leq \vec{m} +
  1$.  If $\max_i \tilde{y}'_i \leq 1/2$ we are done, so suppose that
  $\tilde{y}'_k > 1/2$ for some index $k$.  We claim $m'_k = m_k$,
  i.e.\ site $k$ has not yet jumped.  For the full model, observe that
  the jump response \eqref{eqn:fullresponse} has
  \begin{equation}\label{eqn:sumzero}
    \left[\frac{-2\eta}{1+\eta} + \frac{1-\eta}{1+\eta}
      \frac{2\eta^L}{1-\eta^L}\right] + \sum_{i=1}^{L-1}
    \left[\frac{1-\eta}{1+\eta} \frac{\eta^i +
        \eta^{L-i}}{1-\eta^L}\right] = 0.
  \end{equation}  
  It follows that any particle which has jumped has well coordinate at
  most what it was after (A3), namely $1/2$.  For the toy model, a
  site which has jumped once with neighbors which have each jumped at
  most once has no increase beyond its value after (A3).  In either
  case, $m'_k = m_k$ follows, and $\vec{m}' + \vec{\delta}_k \leq
  \vec{m} + 1$.

  For (iii), if a given site $i$ has not jumped, then $\tilde{y}^*_i$
  is obtained from $\tilde{y}_i > -1/2$ by translating upward $(F^* -
  F)/\lambda$ and then adding the (positive) effects of jumps at the
  other sites.  So there is nothing to check unless the site $i$ has
  jumped.  In this case,
  \begin{equation}\label{eqn:bottomedgejumped}
    \tilde{y}_i + \frac{F^* - F}{\lambda} + a > +\frac{1}{2}
  \end{equation}
  with $a$ the (positive) effect of the jumps at other sites which
  have preceded the jump at $i$, and
  \begin{align*}
    \tilde{y}^*_i &= \tilde{y}_i + \frac{F^* - F}{\lambda} + a -
    \left(\frac{2 \eta}{1 + \eta} - O(\eta^L)\right) \\
    &> \frac{1}{2} - \frac{2\eta}{1+\eta} \\
    &= \frac{F^* - F}{\lambda} + \frac{1}{2} - \frac{F^* - F}{\lambda}
    - \frac{2\eta}{1+\eta} .
  \end{align*}
  So we require
  \begin{equation}
    \frac{F^* - F}{\lambda} + \frac{2\eta}{1 + \eta} < 1.
  \end{equation}
  Since $F \geq 0$, one can easily verify from \eqref{eqn:tildey} that
  the sum of the well coordinates $\tilde{y}_i$ is nonnegative
  \emph{regardless} of $\vec{m}$.  Thus values of $F^*$ with $(F^* -
  F)/\lambda > 1/2$ can correspond only to the sliding state, and (as
  the avalanche algorithm produces only static configurations) we may
  restrict ourselves to $(F^* - F)/\lambda \leq 1/2$.  The choice
  $\eta < 1/3$ makes $2\eta/(1+\eta) < 1/2$.  The desired inequality
  follows.
\end{proof}

\begin{proof}[Lemma \ref{lem:ZFAnoncrossing}]
  Suppose we are applying ZFA to $\vec{m}^1$ and that $\vec{m}'$ is
  either equal to $\vec{m}^1$ or an intermediate configuration
  obtained after some execution of (ZFA3) for which $\vec{m}^1 \leq
  \vec{m}' \leq \vec{m}^2$.  For any $j$ such that $m'_j = m^2_j$,
  \begin{align*}
    \tilde{y}'_j &= \frac{1-\eta}{1+\eta} \sum_{i\in \Z}
    \eta^{|i-j|} (m'_i - m'_j + \alpha_i - \alpha_j) \\
    &\leq \frac{1-\eta}{1+\eta} \sum_{i\in \Z} \eta^{|i-j|} (m^{2}_i -
    m^{2}_j + \alpha_i - \alpha_j) = \tilde{y}^2_j
  \end{align*}
  in the case of the full model, and
  \begin{align*}
    \tilde{y}'_j &= \eta(m'_{j-1} - 2 m'_j + m'_{j+1}) \\
    &\leq \eta(m^2_{j-1} - 2m^2_j + m^2_{j+1}) = \tilde{y}^2_j
  \end{align*}
  for the toy model.

  If (i) $\max_i \tilde{y}^1_i > \max_i \tilde{y}^2_i$, then
  $\tilde{y}'_j < \max_i \tilde{y}^1_i$, and site $j$ will not jump.
  Thus the next iteration of (ZFA3), if any, will produce $\vec{m}''$
  which still has $\vec{m}'' \leq \vec{m}^2$.

  If (ii) $\max_i \tilde{y}^1_i = \max_i \tilde{y}^2_i$, then
  $\tilde{y}'_j \leq \max \tilde{y}^1_i$ and site $j$ will only jump
  if $\vec{m}' = \vec{m}^1$, i.e.\ in (ZFA1), and $j = \argmax_i
  \tilde{y}^1_i$.  If $m^1_j < m^2_j$, this jump does not cause a
  crossing.

  Since $\vec{m}^2 \leq \vec{m}^{2*}$ and $\max_i \tilde{y}^2_i \geq
  \tilde{y}^{2*}_i$ trivially, and having established (i) and (ii),
  for (iii) we need only consider the case where
  \begin{equation}
    \max_i \tilde{y}^1_i = \max_i \tilde{y}^{2*}_i = \max_i \tilde{y}^2_i
  \end{equation}
  and $m^1_j = m^2_j$ for $j = \argmax_i \tilde{y}^1_i$.  As in the
  proof of (i), we find $\tilde{y}^1_j \leq \tilde{y}^2_j$, so $j =
  \argmax_i \tilde{y}^2_i$ as well, so that $m^{2*}_j = m^2_j + 1 >
  m^1_j$.  Invoking (ii), we are done.
\end{proof}

\begin{proof}[Proposition \ref{prop:ZFAthresh}]
  (i) For existence, recall from \eqref{eqn:tildey} that the well
  coordinates can be expressed in terms of the Laplacians $\Delta
  \vec{m}$ and $\Delta \vec{\alpha}$.  We know that $\sum_{i=0}^{L-1}
  \Delta m_i = 0$, so large negative values of $\Delta m_i$ will
  require also large positive values elsewhere.  The equation
  \eqref{eqn:zeroforce1} (at $F = 0$) can be rewritten as
  \begin{equation}
    \lambda \tilde{y}_i = \Delta m_i + \Delta \alpha_i + \Delta \tilde{y}_i.
  \end{equation}
  Noting that $|\Delta \alpha_i|$ and $|\Delta \tilde{y}_i|$ are
  bounded by $2$, we see that large positive values of $\Delta m_i$
  will cause large positive values of $\tilde{y}_i$.  We may therefore
  optimize over $\Delta \vec{m}$ uniformly bounded above by $8$ (since
  anything above this is guaranteed to be worse than taking $\vec{m} =
  0$) and thus below by $8L$, and there are only finitely many
  possibilities.  Existence is immediate.

  Uniqueness requires separate arguments for the full model and the
  toy model.  For the full model, suppose we have $\vec{m}^1$ and
  $\vec{m}^2$ threshold configurations which do not differ by simple
  translation.  Since overall translation does not affect Laplacians,
  it doesn't affect well coordinates, so we may as well assume $\min_i
  m^2_i - m^1_i = 0$ and $\vec{m}^1 \neq \vec{m}^2$.

  We first argue that it suffices to consider $\vec{m}^2 \leq
  \vec{m}^1 + \vec{1}$.  Apply the ZFA to $\vec{m}^1$, producing
  $\vec{m}^{1*}$, also a threshold configuration, and $\vec{m}^{1*}
  \leq \vec{m}^1 + \vec{1}$.  Write $j_1 = \argmax_i \tilde{y}^1_i$.
  If $m^1_j < m^2_j$, we have $\vec{m}^{1*} \leq \vec{m}^2$ by Lemma
  \ref{lem:ZFAnoncrossing}.  We rule out $m^1_j = m^2_j$ because it
  forces
  \begin{equation}
    \tilde{y}^2_j > \tilde{y}^1_j = \max_i \tilde{y}^1_i,
  \end{equation}
  in which case $\vec{m}^2$ is not a threshold configuration.  Thus
  $\vec{m}^{1*}$ is a threshold configuration which has $\vec{m}^1
  \leq \vec{m}^{1*} \leq \vec{m}^1 + \vec{1}$ and $m^1_{j_1} <
  m^{1*}_{j_1}$.  We may as well assume that $\vec{m}^2$ has these
  properties.

  Let $j_2 = \argmax_i \tilde{y}^2_i$.  Since $m^1_{j_1} < m^2_{j_1}$
  and $\vec{m}^2 \leq \vec{m}^1 + \vec{1}$, we have $\tilde{y}^1_{j_1}
  > \tilde{y}^2_{j_1}$, so $j_2 \neq j_1$.  Now consider the
  underlying randomness $\vec{\alpha}$: for $\tilde{y}^1_{j_1} =
  \tilde{y}^2_{j_2}$ we must have, using \eqref{eqn:tildey},
  \begin{equation}
    \sum_{i\in\Z} (\eta^{|i-j_2|} -
    \eta^{|i-j_1|}) \Delta \alpha_i = \sum_{i\in\Z}
    \eta^{|i-j_1|} \Delta m^1_i - \eta^{|i-j_2|} \Delta m^2_i.
  \end{equation}
  Recalling that we're dealing with a periodic system, so that the
  above may be replaced with a finite sum, and that for threshold
  configurations the number of possible values for $\Delta \vec{m}$ is
  finite, we see that $\tilde{y}^1_{j_1} = \tilde{y}^2_{j_2}$ requires
  that a nondegenerate linear functional of $\vec{\alpha}$ takes one
  of finitely many values, which happens with probability zero.

  We turn to uniqueness for the toy model.  Again take threshold
  configurations $\vec{m}^1$ and $\vec{m}^2$ with $\vec{m}^1 \neq
  \vec{m}^2$ and $\min_i m^2_i - m^1_i = 0$.  As we did for the full
  model, we begin by reducing the class of $\vec{m}^2$ we must
  consider.  Write $\vec{m}^{1*}$ for the result of the ZFA applied to
  $\vec{m}^1$.  If $m^1_j < m^2_j$, then $\vec{m}^{1*} \leq
  \vec{m}^2$, as desired.  On the other hand, $m^1_j = m^2_j$ leads to
  a contradiction: let $\ell$ and $r$ be the first indices to the left
  and right, respectively, of $j$ for which $m^2_\ell > m^1_\ell$ and
  $m^2_r > m^1_r$.  Using the formula $z_i = \Delta m_i + \Delta
  \alpha_i$, we see that
  \begin{equation}\label{eqn:edgesbetter}
    z^2_{\ell+1} \geq z^1_{\ell+1} + 1 \quad \text{and}
    \quad z^2_{r-1} \geq z^1_{r-1} + 1,
  \end{equation}
  and it follows that $z^1_{\ell+1} + 1$ and $z^1_{r-1} + 1$ are both
  less than $\max_i z^1_i$.  For reasons as in the uniqueness argument
  for the full model, this inequality is almost surely strict.  Define
  $\vec{m}'$ by
  \begin{equation}
    m'_i = m^1_i + (i-\ell-1)_+ -(i-j)_+ -(i-j+r-\ell)_+ + (i-r-1)_+.
  \end{equation}
  Then $\vec{z}'$ differs from $\vec{z}^1$ in only four locations,
  $\ell + 1$, $j$, $j + r - \ell$, and $r - 1$, with $z'_i - z^1_i$
  having values $+1$, $-1$, $-1$, and $+1$, respectively.  Then
  $\max_i z'_i < \max_i z^1_i$ follows from \eqref{eqn:edgesbetter},
  which is a contradiction.

  Thus it suffices to take $\vec{m}^2 = \vec{m}^{1*}$, and show that
  the assumption $\min_i m^{1*}_i - m^1_i = 0$ leads to a
  contradiction.  When we apply the ZFA to $\vec{m}^1$, the site $j =
  \argmax_i m^1_i$ will jump.  If $\min_i m^{1*}_i - m^1_i = 0$, not
  all sites jump.  Letting $\ell$ and $r$ be as above, we can again
  construct $\vec{m}'$ with $\max_i z'_i < \max_i z^1_i$,
  contradicting optimality and finishing the proof of uniqueness.

  For (ii), take a threshold configuration $\vec{m}^+$ and translate
  it so that $\min_i m^+_i = 0$.  Starting with $\vec{m} = 0$,
  repeatedly apply ZFA.  By Lemma \ref{lem:ZFAnoncrossing}, the
  sequence of $\vec{m}$ produced cannot cross $\vec{m}^+$ unless we
  obtain $\vec{m}$ so that $\max_i \tilde{y}_i = \max_i
  \tilde{y}^+_i$, that is, another threshold configuration.  On the
  other hand, we must jump at least once with each ZFA application, so
  crossing $\vec{m}^+$ after finitely many steps is unavoidable.

  Part (iii) is immediate from (i) and Proposition
  \ref{prop:avalanche}.
\end{proof}

To verify that the description of the threshold configuration given by
\eqref{eqn:toythresh} in Theorem \ref{thm:toythresh} gives a
legitimate vector $\vec{m}^+$ of well numbers, we require the
following elementary lemma.

\begin{lemma}\label{lem:invertibility}
  A vector $\vec{\ell} \in \Z^L$ is equal to $\Delta \vec{m}$ for some
  $\vec{m} \in \Z^L$ if and only if both of the following hold:
  \begin{itemize}
  \item[(i)] $\sum_{i=0}^{L-1} \ell_i = 0$
  \item[(ii)] $\sum_{i=0}^{L-1} i\ell_i \equiv 0 \pmod{L}$
  \end{itemize}
\end{lemma}

\begin{proof}
  That $\Delta$ on $\Q^L$ with periodic boundary is self-adjoint,
  together with standard linear algebra (namely the identification of
  the cokernel with the orthogonal complement of the range) shows that
  condition (i) is necessary and sufficient for $\Delta \vec{m} =
  \vec{\ell}$ to have a solution $\vec{m} \in \Q^L$.  The only
  question is whether there is a solution with integer entries.  For
  this it is necessary and sufficient that a solution $\vec{m} \in
  \Q^L$ have $m_1 - m_0 \in \Z$.  Necessity is obvious and sufficiency
  follows if we set $m_0 = 0$, $m_1$ according to the known difference
  $m_1 - m_0$, and repeatedly use $m_{i+1} = -m_{i-1} + 2 m_i +
  \ell_i$ to obtain the other entries, which will be integers.

  An easy induction shows that for $k \geq 2$,
  \begin{equation}
    m_k = -(k-1) m_0 + k m_1 + \sum_{i=1}^k (k - i) \ell_i.
  \end{equation}
  Setting $k = L$ in the above, recalling $m_0 = m_L$, and rearranging
  we find
  \begin{equation}
    L(m_0 - m_1) = \sum_{i=1}^L (L-i) \ell_i.
  \end{equation}
  From this, we see $m_0 - m_1 \in \Z$ if and only if $\sum_{i=1}^L
  (L-i) \ell_i$ is a multiple of $L$, which is easily shown to be
  equivalent to (ii).
\end{proof}

\begin{proof}[Theorem \ref{thm:toythresh}]
  Lemma \ref{lem:invertibility} guarantees that the specification
  given for $\Delta \vec{m}^+$ is admissible, i.e.\ can be inverted to
  obtain $\vec{m}^+ \in \Z^L$.  To verify the optimality of
  $\vec{m}^+$, we invoke Proposition \ref{prop:ZFAthresh}, claiming
  that the ZFA applied to $\vec{m}^+$ produces $\vec{m}^+ + \vec{1}$.

  We claim that $z^+_i + 1 > z^+_{\max}$ for all $i \neq k^+$ and
  $z^+_{k^+} + 2 > z^+_{\max}$.  Since a jump at site $i$ increases
  $z_{i\pm 1}$ by $1$, each jump that occurs, starting at $\argmax_i
  z^+_i$, causes both its neighbors to jump except possibly if one of
  those neighbors is site $k^+$.  Due to periodicity, both $k^+ \pm 1$
  will jump, increasing $z^+_{k^+}$ by $2$, and it must jump as well.
  Verifying the claim will prove the theorem.

  Using the notation of \eqref{eqn:omega},
  \begin{equation}
    z^+_i = \omega_i +
    \begin{cases}
      \sum_{j=0}^S \delta_{i\sigma(j)} - \delta_{ik^+} &\text{if } S \geq 0 \\ \\
      -\sum_{j=1}^{|S|-1} \delta_{i\sigma(L-j)} - \delta_{ik^+}
      &\text{if } S < 0
    \end{cases}
  \end{equation}
  with all $\omega_i \in (-\frac{1}{2},+\frac{1}{2})$.  Suppose $S >
  0$.  If $i \in \{\sigma(0),\ldots,\sigma(S)\} \setminus \{k^+\}$,
  then
  \begin{equation}
    z^+_i + 1 = (\omega_i + 1) + 1 > z^+_{\max},
  \end{equation}
  and if $i \in \{\sigma(S+1),\ldots,\sigma(L-1)\} \setminus \{k^+\}$,
  then
  \begin{equation}
    z^+_i + 1 = \omega_i + 1 > z^+_{\sigma(S-1)} \vee z^+_{\sigma(S)} = z^+_{\max}.
  \end{equation}
  If $k^+ \in \{\sigma(0),\ldots,\sigma(S)\}$ then
  \begin{equation}
    z^+_{k^+} + 2 = \omega_{k^+} + 2 > z^+_{\sigma(S-1)} \vee z^+_{\sigma(S)} = z^+_{\max},
  \end{equation}
  but if $k^+ \in \{\sigma(S+1),\ldots,\sigma(L-1)\}$ then
  \begin{equation}
    z^+_{k^+} + 2 = \omega_{k^+} + 1 > z^+_{\sigma(S-1)} \vee z^+_{\sigma(S)} = z^+_{\max}.
  \end{equation}
  We omit the verification in the cases $S = 0$ and $S < 0$, these
  being similar exercises in checking cases.
\end{proof}

\begin{proof}[Proposition \ref{prop:avalanchewaves}]
  Note first that \eqref{eqn:trapezoid} follows once we've established
  the corresponding changes to $\vec{z}$, as the change in $\vec{m}$
  has the correct Laplacian, and minimum $0$.
  
  We argue by induction on $\ell$.  When $\ell = 1$, we are only
  characterizing the result of a single iteration of the ZFA.  It is
  straightforward to verify that the set of sites that jumps is $W_1$:
  the site $i$ itself makes the first jump, increasing the height of
  each neighbor $i \pm 1$ by one, these will jump if and only if
  $z_{i\pm 1} + 1 > z_i$.  This outward moving wave terminates when
  the sites $i_L+1$ and $i_R-1$ have jumped, as the additions to their
  neighbors to the left and right, respectively, are by definition not
  sufficient to be force these to jump.  We see immediately that the
  number of jumps is
  \begin{equation*}
    (i_R - 1) - (i_L + 1) + 1 = i_R - i_L - 1.
  \end{equation*}
  Since $\ell = 1$, we know that either $i-i_L = 1$ or $i_R-i = 1$; by
  symmetry, we may as well assume the former.  Then
  \begin{equation*}
    (i-i_L)(i_R - i) = i_R - i = i_R - (i_L + 1),
  \end{equation*}
  and (i) is satisfied.  For (ii), we observe that the result of an
  interval of sites jumping is as follows:
  \begin{equation*}
    \begin{array}{lrrcrrrrrrcrr}
      &        &   &i_L&   &   &   &        &   &   &i_R \\ \hline
      \text{jumps}: & \cdots & 0 & 0 & 1 & 1 & 1 & \cdots & 1 & 1 & 0 & 0 &
      \cdots \\
      \text{change in $\vec{z}$}: & \cdots & 0 & 1 & -1 & 0 & 0
      & \cdots & 0 & -1 & 1 & 0 & \cdots
    \end{array}
  \end{equation*}
  We see that $z_{i_L} \to z_{i_L} + 1$ and $z_{i_R} \to z_{i_R} + 1$,
  as needed.  Again assuming that $i - i_L = 1$, we find $z_i \to z_i
  - 1$.  Since we have
  \begin{equation*}
    i_L + i_R - i = i_R - 1,
  \end{equation*}
  so $z_{i_L+i_R-i} \to z_{i_L + i_R - i} - 1$.
  
  Now consider general $\ell > 1$, assuming the result holds for
  smaller values.  Following the same reasoning as above, the sites
  $W_1$ jump in the first application of the ZFA, but now $i - i_L >
  1$ and $i_R - i > 1$, so both $i\pm 1$ jump, and the value of $z_i$
  is unchanged, hence still the maximum.  Regarding the configuration
  after jumping sites $W_1$ as the new starting point, we have the
  same $i$ and $z_i$, and $i_L' = i_L + 1$ and $i_R' = i_R - 1$.
  Applying the inductive hypothesis, we know the effect of iterations
  $2,\ldots,\ell$.  The number of jumps is therefore
  \begin{equation*}
    [(i_R - 1) - (i_L + 1) + 1] + (i - i_L')(i_R' - i) = (i - i_L)(i_R - i).
  \end{equation*}
  The changes in $\vec{z}$ for iteration $1$ consist of
  \begin{align}
    z_{i_L} &\to z_{i_L} + 1 & z_{i_R} &\to z_{i_R} +
    1 \label{eqn:outerplus} \\
    z_{i_L+1} &\to z_{i_L+1} - 1 & z_{i_R-1} &\to z_{i_R-1} -
    1. \label{eqn:outerminus}
  \end{align}
  The changes due to iterations $2,\ldots,\ell$ are
  \begin{align}
    z_{i_L+1} &\to z_{i_L+1} + 1 & z_{i_R-1} &\to z_{i_R-1} + 1 \label{eqn:innerplus}\\
    z_i &\to z_i - 1 & z_{i_L'+i_R'-i} &\to z_{i_L'+i_R'-i} -
    1. \label{eqn:innerminus}
  \end{align}
  Noting $i_L'+i_R'-i = i_L + i_R - i$, we see that
  \eqref{eqn:outerplus} and \eqref{eqn:innerminus} are the desired
  changes and that \eqref{eqn:outerminus} and \eqref{eqn:innerplus}
  cancel.
\end{proof}

\begin{proof}[Proposition \ref{prop:iidcond}]
  We begin by describing the distribution claimed for $\vec{\omega}$
  in greater detail.  For the uniform surface measure on the
  intersection of the cube $(-\frac{1}{2},+\frac{1}{2})^L$ with the
  planes $\sum_{n=0}^{L-1} b_n \in \Z$, a consequence of $|b_n| <
  \frac{1}{2}$ is that this surface can be recognized as the graph of
  a function:
  \begin{equation}\label{eqn:BetaSurface}
    b_{L-1} = g(b_0,\ldots,b_{L-2}) \equiv 
    \left\llbracket \textstyle\sum_{n=0}^{L-2} b_n \right\rrbracket 
    - \textstyle\sum_{n=0}^{L-2} b_n
  \end{equation}
  is immediate from $b_{L-1} + \sum_{n=0}^{L-2} b_n = \left\llbracket
    \sum_{n=0}^{L-2} b_n \right\rrbracket$, which is forced since the
  left-hand side is exactly an integer, and since $|b_{L-1}| <
  \frac{1}{2}$, it must be the integer nearest $\sum_{n=0}^{L-2} b_n$.
  The function $g$ has constant gradient $(-1,\ldots,-1)$ where the
  gradient exists, and it fails to exist only on the
  $(L-2)$-dimensional set
  \[
  \left\{(b_0,\ldots,b_{L-2}) : \sum_{n=0}^{L-2} b_n \in \frac{1}{2} +
    \Z \right\}.
  \]
  We therefore recognize the law of $\{\beta_n\}_{n=0}^{L-1}$ as the
  result of taking $\{\beta_n\}_{n=0}^{L-2}$ to be i.i.d.\ uniform and
  pushing this measure forward onto the graph of $g$.  This
  facilitates the following calculation, for trigonometric polynomials
  $f_n(t) = \sum_{|k| \leq K} \hat{f}_n(k) \exp(2\pi i k t)$, $K$ an
  arbitrary positive integer:
  \begin{equation}
    \bbE \prod_{n=0}^{L-1} f_n(\beta_n)
    = \sum_{|k_n| \leq K} \left(\prod_{n=0}^{L-1} \hat{f}_n(k_n) \right) \bbE \exp[2\pi i k \cdot \beta],
  \end{equation}
  and
  \begin{align*}
    \bbE &\exp[2\pi i k \cdot \beta] \\
    &= \int_{(-\frac{1}{2},+\frac{1}{2})^{L-1}} \exp\left\{2\pi i \left[\sum_{n=0}^{L-2} k_n b_n + k_{L-1} \left(\llbracket \sum_{n=0}^{L-2} b_n \rrbracket - \sum_{n=0}^{L-2} b_n\right) \right] \right\} \, db_0 \cdots db_{L-2} \\
    &= \int_{(-\frac{1}{2},+\frac{1}{2})^{L-1}} \exp\left\{2\pi i \left[\sum_{n=0}^{L-2} (k_n - k_{L-1}) b_n + k_{L-1} \llbracket \sum_{n=0}^{L-2} b_n \rrbracket \right] \right\} \, db_0 \cdots db_{L-2} \\
    &= \int_{(-\frac{1}{2},+\frac{1}{2})^{L-1}} \exp\left\{2\pi i \left[\sum_{n=0}^{L-2} (k_n - k_{L-1}) b_n \right] \right\} \, db_0 \cdots db_{L-2} \\
    &= \mathbf{1}(k_0 = \cdots = k_{L-1}).
  \end{align*}

  Recall that $\omega_i$ is the representative in
  $\left(-\frac{1}{2},+\frac{1}{2}\right)$ of the equivalence class of
  $\Delta \alpha_i \pmod{1}$, so it will suffice to understand the law
  of 1-periodic functions of $\{\Delta \alpha_i\}$.  With
  trigonometric polynomials $f_n$ as before, we compute
  \begin{equation}\label{eqn:ExpTrigFinite}
    \bbE \prod_{n=0}^{L-1} f_n(\Delta \alpha_n) = \sum_{|k_n| \leq K} \left(\prod_{n=0}^{L-1} \hat{f}_n(k_n)\right) \bbE \exp[2\pi i k \cdot \Delta \alpha],
  \end{equation}
  the summation over integer vectors $\vec{k}$ with all components
  bounded by $K$, and
  \begin{align*}
    \bbE \exp[2\pi i k \cdot \Delta \alpha]
    &= \bbE \exp[2\pi i \Delta k \cdot \alpha] = \prod_{n=0}^{L-1} \bbE \exp[2\pi i \Delta k_n \alpha_n] \\
    &= \mathbf{1}(\Delta k = 0) = \mathbf{1}(k_0 = \cdots = k_{L-1}),
  \end{align*}
  since the kernel of the periodic Laplacian consists of constant
  vectors.  Then \eqref{eqn:ExpTrigFinite} simplifies as
  \begin{equation}
    \bbE \prod_{n=0}^{L-1} f_n(\Delta \alpha_n) = \sum_{|k|
      \leq K} \prod_{n = 0}^{L-1} \hat{f}_n(k)
  \end{equation}
  where $k$ is now a single integer (corresponding to a vector with
  components $k_n$ which are identical).
  
  Thus
  \begin{equation}
    \bbE \prod_{n=0}^{L-1} f_n(\beta_n) = \sum_{|k| \leq K} \prod_{n=0}^{L-1} \hat{f}_n(k) = \bbE \prod_{n=0}^{L-1} f_n(\Delta \alpha_n),
  \end{equation}
  and by Stone-Weierstrass we extend to general $1$-periodic functions
  $f_n$ as needed to verify the proposition.
\end{proof}

\begin{proof}[Corollary \ref{cor:cltforS}]
  Using Proposition \ref{prop:iidcond}, we have
  \begin{equation}
    S = \sum_{i=0}^{L-1} \llbracket \Delta \alpha_i
    \rrbracket = \sum_{i=0}^{L-1} \Delta \alpha_i - \omega_i
    = -\sum_{i=0}^{L-1} \omega_i
  \end{equation}
  for $\omega_0,\ldots,\omega_{L-2}$ \iid with mean 0 and variance
  $\frac{1}{12}$ and $|\omega_{L-1}| < \frac{1}{2}$.  The standard
  central limit theorem then gives (i).  The number of topological
  defects is one of $|S|$ or $|S|+2$, and (ii) is immediate.
\end{proof}

The exchangeability claimed in Theorem \ref{thm:exchangeable} requires
a more detailed examination of the threshold configuration.  We begin
by noting the formula for $\Delta \vec{m}$ at $(\pm)$-threshold
\eqref{eqn:toythresh} can be viewed as a result of applying two
corrections to the $-\llbracket \Delta \vec{\alpha} \rrbracket$
sequence:
\begin{equation}\label{eqn:twocorrections}
  \Delta m_i = -\llbracket \Delta \alpha_i \rrbracket + J'_i + (\delta_{i\ell^+} - \delta_{i\ell^-}),
\end{equation}
where
\begin{equation}\label{eqn:Jprime}
  J'_i =
  \begin{cases}
    -\mathbf{1}(i \in \sigma\{L-|S|,\ldots,L-1\}) & \text{if } S < 0 \\
    0 & \text{if } S = 0 \\
    \mathbf{1}(i \in \sigma\{0,\ldots,S-1\}) & \text{if } S > 0
  \end{cases}
\end{equation}
and $\ell^\pm$ are selected as follows: for the $(+)$-threshold
configuration, we set
\begin{align}
  \label{eqn:ellplus}
  \ell^+ &=
  \begin{cases}
    \sigma(L - |S|) & \text{if } S < 0 \\
    \sigma(0) & \text{if } S = 0 \\
    \sigma(S) &\text{if } S > 0
  \end{cases} \\
  \intertext{and for the $(-)$-threshold configuration, we set}
  \label{eqn:ellminus}
  \ell^- &=
  \begin{cases}
    \sigma(L - |S| - 1) & \text{if } S < 0 \\
    \sigma(L - 1) & \text{if } S = 0 \\
    \sigma(S - 1) & \text{if } S > 0.
  \end{cases}
\end{align}
In both cases, the choice of $\ell^\pm$ dictates a corresponding
$\ell^\mp$ via the $L$-divisibility condition of Lemma
\ref{lem:invertibility}. We thus view the $(\pm)$-threshold
configurations as ``one up, one down'' perturbations of $-\llbracket
\Delta \alpha \rrbracket + J'$, with the same spacing
\begin{equation}\label{eqn:d}
  d \equiv \ell^+ - \ell^- \pmod{L} = \sum_{i=0}^{L-1}
  i(\llbracket \Delta \alpha_i \rrbracket - J'_i) \pmod{L}
\end{equation}
between the $\pm 1$, and we insist on choosing $\ell^\pm$ for the
$(\pm)$-threshold, respectively.

The location of the negative defect in the $(-)$-threshold is
important for the threshold-to-threshold evolution, and, in light of
the above, this amounts to understanding $d$ and $\sigma$.  For this,
and the exchangeability result Theorem \ref{thm:exchangeable}, we need
to understand the relationship between $d$ and $\vec{\omega}$.
Fortunately these interact as nicely as one could hope.

\begin{lemma}\label{lem:dindependent}
  The difference $d$ between $\ell^\pm$ defined by \eqref{eqn:d} is
  uniform on $\{0,\ldots,L-1\}$ and independent of $\vec{\omega}$.
\end{lemma}
\begin{proof}
  We begin with the part of $d$ which depends on $\llbracket \Delta
  \vec{\alpha} \rrbracket$, claiming that
  \begin{equation}
    \sum_{i=0}^{L-1} i \llbracket \Delta \alpha_i \rrbracket \pmod{L}
  \end{equation}
  is uniform on $\{0,\ldots,L-1\}$ and independent of $\vec{\omega}$.

  For independence from $\vec{\omega}$, it is sufficient to consider
  $\{\omega_i\}_{i=1}^{L-1}$, since $\omega_0$ is a function of these.
  We have
  \begin{equation}
    \sum_{i=0}^{L-1} i \llbracket \Delta \alpha_i \rrbracket
    = \sum_{i=0}^{L-1} i (\Delta \alpha_i - \omega_i)= L(\alpha_0 - \alpha_{L-1}) - \sum_{i=0}^{L-1} i
    \omega_i,
  \end{equation}
  and claim that $\{\alpha_0 - \alpha_{L-1} \mod{1}, \omega_1, \ldots,
  \omega_{L-1}\}$ are distributed as i.i.d.\ uniform (mod 1) variates
  conditioned to have
  \begin{equation}
    L(\alpha_0 - \alpha_{L-1}) - \sum_{i=1}^{L-1} i \omega_i \in \Z.
  \end{equation}
  We calculate in the manner of Proposition \ref{prop:iidcond}.  For
  $f_n(t) = \sum_{|k| \leq K} \hat{f}_n(k) \exp(2\pi i k t)$, consider
  $\bbE f_0(\alpha_0 - \alpha_{L-1}) \prod_{n=1}^{L-1} f_n(\Delta
  \alpha_n)$:
  \begin{equation}
    \sum_{|k_n| \leq K} \prod_{n=0}^{L-1} \hat{f}_n(k_n) \bbE \exp [2\pi i k \cdot(\alpha_0 - \alpha_{L-1}, \Delta \alpha_1, \ldots, \Delta \alpha_{L-1})].
  \end{equation}
  Write $A$ for the matrix mapping $(\alpha_0,\ldots,\alpha_{L-1})
  \mapsto (\alpha_0 - \alpha_{L-1}, \Delta \alpha_1, \ldots, \Delta
  \alpha_{L-1})$.  We need to evaluate
  \begin{equation}
    \bbE \exp [2\pi i \vec{k} \cdot A\vec{\alpha}] = \bbE \exp[2\pi i A^T \vec{k} \cdot \vec{\alpha}] = \mathbf{1}(A^T \vec{k} = 0),
  \end{equation}
  and therefore require a description of $\ker A^T$.  We have
  \begin{equation}
    A =
    \begin{pmatrix}
      1 & 0 & 0 & \cdots & 0 & -1 \\
      1 & -2 & 1 & \cdots & 0 & 0 \\
      0 & 1 & -2 & \cdots & 0 & 0 \\
      \vdots & \vdots & \vdots & \ddots & \vdots & \vdots \\
      0 & 0 & 0 & \cdots & -2 & 1 \\
      1 & 0 & 0 & \cdots & 1 & -2
    \end{pmatrix},
  \end{equation}
  and see that $A^T$ has rows $2$ through $L-2$ (indexing $0$ through
  $L-1$) in common with the Laplacian; that $(\Delta k_2, \ldots,
  \Delta k_{L-2}) = \vec{0}$ means $(k_1,\ldots,k_{L-1})$ is flat, so
  that
  \begin{equation}
    (k_1,\ldots,k_{L-1}) = (an + b)_{n=1}^{L-1}
  \end{equation}
  for some constants $a$ and $b$.  The second row then gives
  \begin{equation}
    0 = -2(1a + b) + 1(2a + b) = -b.
  \end{equation}
  The first row gives
  \begin{equation}
    0 = k_0 + 1a + (L-1)a = k_0 + La,
  \end{equation}
  and the last
  \begin{equation}
    0 = -(-La) + (L-2)a - 2(L-1)a = 0
  \end{equation}
  imposes no additional constraint.  Thus $A^T \vec{k} = \vec{0}$ if
  and only if
  \begin{equation}
    k = (k_0,\ldots,k_{L-1}) = (-La, 1a, 2a, \ldots, (L-1) a)
  \end{equation}
  for some constant $a$.

  Compare this with the following: let $\beta_1,\ldots,\beta_{L-1}$ be
  i.i.d.\ uniform $(-\frac{1}{2},+\frac{1}{2})$, $\theta \in
  \{0,\ldots,L-1\}$ uniform and independent of the $\beta_i$, and
  \begin{equation}
    \gamma =  \frac{1}{L} \left(\theta + \sum_{n=1}^{L-1} n \beta_n \right)  \pmod{1}.
  \end{equation}
  For $f_0,\ldots,f_{L-1}$ as before, we compute
  \begin{equation}
    \bbE f_0(\gamma) \prod_{n=1}^{L-1} f_n(\beta_n) 
    = \sum_{|k_n| \leq K} \prod_{n=0}^{L-1} \hat{f}_n(k_n) \bbE \exp [2\pi i k \cdot (\gamma, \beta_1, \ldots, \beta_{L-1})].
  \end{equation}
  Here
  \begin{align*}
    \bbE \exp &[2\pi i k \cdot (\gamma, \beta_1, \ldots,
    \beta_{L-1})] \\
    &= \bbE \exp \left\{2\pi i \left[\frac{k_0}{L} \left(\theta + \sum_{n=1}^{L-1} n \beta_n\right) + \sum_{n=1}^{L-1} k_n \beta_n\right] \right\} \\
    &= \bbE \exp \left\{2\pi i \left[ \frac{k_0 \theta}{L} + \sum_{n=1}^{L-1} \left(\frac{nk_0}{L} + k_n \right) \beta_n \right]\right\} \\
    &= \left(\frac{1}{L} \sum_{t=0}^{L-1} e^{2\pi i k_0 t/L} \right)
    \bbE \exp \left\{2\pi i \sum_{n=1}^{L-1} \left(\frac{nk_0}{L} +
        k_n \right) \beta_n \right\}.
  \end{align*}
  Note that $e^{2\pi i k_0/L}$ is an $L^{\mathrm{th}}$ root of unity,
  so the left sum above is zero unless $L$ divides $k_0$, in which
  case the sum is $L$.  But if $L$ divides $k_0$, say $k_0 = -La$,
  then
  \begin{equation}
    \bbE \exp \left\{2\pi i \sum_{n=1}^{L-1} \left(\frac{nk_0}{L} + k_n \right) \beta_n \right\} = \mathbf{1}\left(k_n = \frac{-nk_0}{L} \text{ for } n = 1,\ldots,L-1\right),
  \end{equation}
  which can be nonzero only if $k_n = -n(-La)/L = na$ for $n =
  1,\ldots,L-1$.  Thus
  \begin{equation}
    \{\gamma,\beta_1,\ldots,\beta_{L-1}\} \stackrel{d}{=} \{\alpha_0 - \alpha_{L-1} \pmod{1}, \omega_1, \ldots, \omega_{L-1}\}.
  \end{equation}
  
  Now that we know $\sum_{i=0}^{L-1} i \llbracket \Delta \alpha_i
  \rrbracket$ is independent of $\vec{\omega}$, and that $\vec{J}'$ is
  a function of $\vec{\omega}$, we use the following elementary fact:
  if $X$ and $Y$ are independent random variables in $\Z/L\Z$ and $Y$
  is uniform, then $X + Y$ is uniform and independent of $X$.
  Independence of $d$ and $\vec{\omega}$ is immediate.
\end{proof}

\begin{proof}[Theorem \ref{thm:exchangeable}]
  Exchangeability of the components $\omega_i$ is immediate from
  Proposition \ref{prop:iidcond}.  We have
  \begin{equation}\label{eqn:zfromJprime}
    z^+_i = \Delta m_i + \Delta \alpha_i = \omega_i
    + J'_i + \delta_{i\ell^+} - \delta_{i\ell^-}.
  \end{equation}
  By construction \eqref{eqn:Jprime} and \eqref{eqn:ellplus}, $J'_i$
  and $\delta_{i\ell^+}$ are functions of the \emph{value} $\omega_i$
  and the \emph{unordered} set of values
  $\{\omega_0,\ldots,\omega_{L-1}\}$.  Using the preceding Lemma
  \ref{lem:dindependent}, we find $\ell^- = \ell^+ - d$ is uniform on
  $\{0,\ldots,L-1\}$ and independent of $\vec{\omega}$.

  We then recognize $z^+_i$ given by \eqref{eqn:zfromJprime} as a
  function of $\omega_i$, the set of values
  $\{\omega_0,\ldots,\omega_{L-1}\}$, and $\ell^-$, the last of which
  is independent of $\vec{\omega}$.  Exchangeability of the components
  of $\vec{z}^+$ follows.
\end{proof}

\begin{proof}[Corollary \ref{cor:bridge}]
  We first use Theorem \ref{thm:exchangeable} and a standard result
  (see for example \cite[Thm.\ 24.2]{Billingsley} or \cite[Thm.\
  16.23]{Kallenberg}) to show that the processes
  \begin{equation}\label{eqn:sumz}
    \hat{s}^{(L)}(t) \equiv L^{-1/2} \sum_{i=0}^{\lfloor Lt
      \rfloor} z_i^+
    \qquad (0 \leq t \leq 1)
  \end{equation}
  converge in distribution in the Skorokhod space $D([0,1])$ to
  $(12)^{-1/2} B(t)$ where $B(t)$ is standard Brownian bridge.  We
  claim that we have distributional convergence,
  \begin{equation}
    \label{eqn:pairmeas}
    \left(L^{-1/2} \sum_{i=0}^{L-1} z^+_i , L^{-1} \sum_{i=0}^{L-1}
      (z^+_i)^2 \delta_{L^{-1/2} z^+_i}\right)
    \stackrel{d}{\to}
    (0, (12L)^{-1} \delta_0) \in \R \times \mathcal{M}(\R),
  \end{equation}
  where $\mathcal{M}(\R)$ is the space of locally finite measures on
  $\R$ equipped with the vague topology.  In fact, the first component
  is exactly equal to 0, so we focus on the second component, which we
  write as
  \begin{equation}
    \label{eqn:randmeas}
    L^{-1} \sum_{i=0}^{L-1} (z^+_i)^2 \delta_0 + L^{-1}
    \sum_{i=0}^{L-1} (z^+_i)^2 (\delta_{L^{-1/2} z^+_i} - \delta_0).
  \end{equation}
  We claim the second sum above can be ignored as $L \to \infty$.  Fix
  a continuous, compactly supported function $f$ on $\R$, and any
  $\epsilon > 0$.  Choose $L$ sufficiently large that $|x| < L^{-1/2}$
  implies $|f(x) - f(0)| < \epsilon$, and observe that
  \begin{equation}
    \left|\int f(x) L^{-1} \sum_{i=0}^{L-1} (z^+_i)^2 (\delta_{L^{-1/2} z^+_i}
      - \delta_0)(dx) \right|
    \leq \frac{\epsilon}{4}    
  \end{equation}
  almost surely, since $|z^+_i| \leq \frac{1}{2}$.  Distributional
  convergence of the first sum of measures in \eqref{eqn:randmeas}
  amounts to distributional convergence of the coefficient
  \begin{align*}
    L^{-1} \sum_{i=0}^{L-1} (z^+_i)^2
    &= L^{-1} \sum_{i=0}^{L-1} (\omega_i + z^+_i - \omega_i)^2 \\
    &= L^{-1} \sum_{i=0}^{L-1} \omega_i^2 + L^{-1} \sum_{i=0}^{L-1}
    (z^+ - \omega_i) (z^+_i + \omega_i) \stackrel{d}{\to}
    \frac{1}{12}.
  \end{align*}
  Here we have used the (weak) law of large numbers on
  $\sum_{i=0}^{L-2} \omega_i^2$, since removing one term restores
  independence, and
  \begin{align*}
    \left| L^{-1} \sum_{i=0}^{L-1} (z^+ - \omega_i) (z^+ + \omega_i)
    \right| &\leq L^{-1} \sum_{i=0}^{L-1} |J'_i + \delta_{i\ell^+} -
    \delta_{i\ell^-}|(2) \\
    &= 2 L^{-1} \left|\sum_{i=0}^{L-1} \omega_i \right|
    \stackrel{d}{\to} 0
  \end{align*}
  again using law of large numbers.  The convergence
  \eqref{eqn:pairmeas} holds, and scaling limit for $\hat{s}^{(L)}(t)$
  follows.
  
  We now return to $s^{(L)}(t)$.  Writing $\hat{s}_i \equiv
  \sum_{j=0}^i z^+_i$, a routine calculation gives
  \begin{equation}
    s_i - \left(\hat{s}_i - \frac{1}{L} \sum_{j=0}^{L-1}
      \hat{s}_j\right) = \alpha_i - \alpha_{i+1}.
  \end{equation}
  In particular, the difference on the left-hand side is bounded by a
  constant, and thus disappears in the central limit scaling.  Note
  also that
  \begin{equation}
    \frac{1}{L} \sum_{j=0}^{L-1} \hat{s}_j = \int_0^1
    \hat{s}_{\lfloor t/L \rfloor} \, dt,
  \end{equation}
  and that integration $\int_0^1 \cdot \, dt$ is a continuous
  functional on the Skorokhod space $\mathcal{D}([0,1])$.  The
  convergence to the distribution of \eqref{eqn:bridge} follows.

  That $B_0(t)$ has mean zero is immediate, and that it is Gaussian
  follows from easy arguments.  The discrete analogue, a Gaussian
  vector with its sum subtracted from each component, is of course
  standard, since (possibly degenerate) Gaussian distributions are
  preserved under affine maps.  Working on the level of continuous
  processes, we can fix some $0 = t_0 < t_1 < \cdots < t_{n-1} < t_n =
  1$ and observe using standard properties of Brownian bridge that
  \begin{equation}
    \int_0^1 B(r) \, dr - \sum_{i=1}^n \frac{1}{2}[B(t_{i-1}) +
    B(t_i)] (t_i - t_{i-1})
  \end{equation}
  is Gaussian and independent of $(B(t_0),\ldots,B(t_n))$.
  
  Stationarity can be deduced from that of the sequence of strains
  $s_i$, or from computing the covariance $\bbE B_0(t) B_0(t')$ for
  some $t,t' \in [0,1]$ and recognizing this as a function of the
  difference $t' - t$; recall that wide-sense stationarity and
  stationarity are equivalent for Gaussian processes.  The formula
  \eqref{eqn:bridgecov} is obtained using Fubini's theorem and
  calculus.
\end{proof}

\begin{proof}[Proposition \ref{prop:polarization}]
  By the apparent exchangeability of the components of $\vec{\zeta}$,
  the unordered pair $\{\pi_L,\pi_R\}$ is uniformly distributed over
  all pairs of (mod $L$ equivalence classes) of indices.  Using Lemma
  \ref{lem:dindependent} as in the proof of Theorem
  \ref{thm:exchangeable}, we find that $k^-$ ranges over all indices
  and independent of $\vec{\zeta}$, and thus also independent of
  $\{\pi_L, \pi_R\}$.

  As in the discussion surrounding \eqref{eqn:jLRcenter}, we may by
  translation assume that $k^- = 0$, and select the representatives of
  $\pi_L$ and $\pi_R$ which satisfy
  \begin{equation}
    \pi_R - L < \pi_L \leq 0 \leq \pi_R < \pi_L + L.
  \end{equation}
  By the discussion above, we find $(\pi_L,\pi_R)$ is uniformly
  distributed over the set
  \begin{equation}
    \label{eqn:legalindices}
    \{(i,j) \in \Z^2 : i \leq 0, j \geq 0, i \neq j, |i| + j < L\},
  \end{equation}
  and independent of the \emph{unordered} set of values
  $\{\omega_0,\ldots,\omega_{L-1}\}$ or
  $\{\zeta_0,\ldots,\zeta_{L-1}\}$.  The set in
  \eqref{eqn:legalindices} has cardinality $(L+2)(L-1)/2$.

  We observe that if \emph{either} $\pi_L = 0$ or $\pi_R = 0$, the
  threshold-to-threshold polarization is zero because the
  $(\pm)$-threshold configurations are the same.  This follows using
  \eqref{eqn:kplus}, which says that if $k^-$ is one of $\pi_L$ or
  $\pi_R$, then $k^+$ is the other, and the formulas
  \eqref{eqn:zminus} and \eqref{eqn:zplus} match.  Note that these
  cases contribute a bounded (in fact, zero) quantity to the
  polarization, and occur with probability only of order $O(L^{-1})$.
  In the limit as $L \to \infty$, this can be ignored.  We replace
  \eqref{eqn:legalindices} with
  \begin{equation}
    \label{eqn:legalindicesredux}
    \{(i,j) \in \Z^2 : i \leq -1, j \geq 1, |i| + j < L\},
  \end{equation}
  which has cardinality $\binom{L}{2}$.

  As stated, we will now make the \emph{assumption} that the genuine
  situation can be approximated by $\zeta_{\pi(0)} = \zeta_{\pi(1)} =
  -1/2$ and that the sequences $\cJ_L$ and $\cJ_R$ of
  \eqref{eqn:sequencesL} and \eqref{eqn:sequencesR} can be
  approximated by i.i.d.~uniform $(-\frac{1}{2},+\frac{1}{2})$
  variables independent of $\pi(0)$ and $\pi(L)$.  Following this
  assumption, the calculation is exact.  Subtracting $-1/2$ from all
  of these, we obtain i.i.d.~uniform $(0,1)$ variables
  \begin{equation}
    X_{\pi_L + 1}, \ldots, X_0, \ldots, X_{\pi_R - 1}.
  \end{equation}
  We might extend this to a bi-infinite sequence of i.i.d.~uniform
  $(0,1)$ variates, and then define for $0 \leq x \leq 1$
  \begin{subequations}
    \begin{align}
      j_L'(x) &= \max \{i \leq 0 : X_i \leq x\}
      & j_R'(x) &= \min \{j \geq 0 : X_j \leq x\}   \\
      j_L(x) &= j_L'(x) \vee \pi_L & j_R(x) &= j_R'(x) \wedge
      \pi_R. \label{eqn:jprocdef}
    \end{align}
  \end{subequations}
  Recalling the cumulative avalanche size $\Sigma(x)$ and polarization
  $P(x)$ are given by $L P(x) = \Sigma(x) = -j_L(x) j_R(x)$, we wish
  to characterize the distribution of the pair $(j_L(x),j_R(x))$.

  The distribution of $(j_L(x),j_R(x))$ can be computed precisely on
  the discrete level, but since we are interested in the behavior as
  $L \to \infty$, we may as well rescale and pass to continuous
  variates.  We claim that as $L \to \infty$, for fixed $u \geq 0$,
  \begin{equation}
    \left(\frac{-j_L'(u/L)}{L},\frac{j_R'(u/L)}{L},\frac{-\pi_L}{L},\frac{\pi_R}{L}\right)
    \stackrel{d}{\to} (\gamma'_L(u),\gamma'_R(u),\rho_L,\rho_R),
  \end{equation}
  where $\gamma'_L(u)$ and $\gamma'_R(u)$ are independent exponential
  random variables with rate $u$, and are independent from the pair
  $(\rho_L,\rho_R)$, which is uniformly distributed on the triangular
  region with vertices $(0,0)$, $(1,0)$, and $(0,1)$.

  To see this, note first that $j_L'(x)$ and $j_R'(x)$ are
  conditionally independent on the event $j_R'(x) > 0$.  For fixed
  $u$, as $L \to \infty$, the probability that $j_R'(u/L) = 0$ tends
  to zero.  Observe also that for fixed $u$,
  \begin{equation}
    \bbP\left(\frac{j_R'(u/L)}{L} \leq t\right) = \sum_{n=0}^{\lfloor
      Lt \rfloor} \left(1 - \frac{u}{L}\right)^n \frac{u}{L} = 1 -
    \left(1 - \frac{u}{L}\right)^{\lfloor Lt \rfloor + 1} \to 1 - e^{-ut},
  \end{equation}
  pointwise for all $t$.  Finally, that $(-\pi_L/L,\pi_R/L)$ converges
  distributionally to $(\rho_L,\rho_R)$ is immediate, since computing
  an expectation of some function with respect to the law of
  $(-\pi_L/L, \pi_R/L)$ is more or less a Riemann sum for the integral
  over the triangle.

  Since $(a,b,c,d) \mapsto (a \wedge c, b \wedge d)$ is continuous,
  and distributional convergence is preserved under continuous maps,
  it follows that
  \begin{multline*}
    \left(\frac{-j_L(u/L)}{L},\frac{j_R(u/L)}{L}\right) =
    \left(\frac{-j_L'(u/L)}{L} \wedge \frac{-\pi_L}{L},
      \frac{j_R'(u/L)}{L} \wedge \frac{\pi_R}{L}\right) \\
    \stackrel{d}{\to} (\gamma_L'(u) \wedge \rho_L, \gamma_R'(u) \wedge
    \rho_R) \equiv (\gamma_L(u), \gamma_R(u)).
  \end{multline*}
  We address the limiting statistics of the former using a calculation
  with the latter, continuous variates.

  Rescaling $\Sigma(x)$ as
  \begin{equation}
    \varsigma(u) \equiv \lim_{L \to \infty} \Sigma(u/L)/L^2,
  \end{equation}
  we obtain the density of $\varsigma(u)$, which we denote $p_u(s) =
  \bbP( \varsigma(u) \in ds)/ds$,
  \begin{equation}
    p_u(s) =  \int_{0}^{1} {\rm d}x \int_{0}^{1 - x} {\rm d}y \, \delta(s - xy) e^{-u(x+y)} 
    \left [ 2 + 4u (1 - x - y) + u^2  ( 1 - x - y )^2 \right ],
  \end{equation}
  where $\delta(x)$ is the Dirac delta. Carrying out next the
  integration over $y$, making then a change of variable $z = x + s/x$
  in the remaining integral, and taking care of the integration
  boundaries we obtain
  \begin{equation}
    p_u(s) 
    = \int_{2\sqrt{s}}^1 {\rm d} z \, e^{-zu} \frac{4+8u(1 - z)+ 2u^2(1-z)^2}{(z^2 - 4s)^{1/2}},
  \end{equation}
  which is supported on $[0,\frac{1}{4}]$.  This is the formula
  claimed in \eqref{eqn:pdist}. The expectation value of
  $\varsigma(u)$, \eqref{eqn:polarscaling}, then follows from the
  distribution $p_u(s)$ by an exchange of the order of integration and
  some repeated integration by parts.
\end{proof}

\section{Conclusions and remaining questions}\label{sec:conclusion}

The CDW toy model introduced in \cite{KMShort13} and studied in this
article exhibits a critical depinning transition.  It retains
similarities with the untruncated CDW model, while admitting some
explicit formulas which make rigorous analysis possible.  However, it
does not appear to be completely trivial.  Our understanding of the
threshold-to-threshold evolution is rather complete, as the changes
are confined to a single active region growing in a simple way, but
the flat-to-threshold evolution has so far resisted nice analytical
characterizations.  In simulations we see multiple regions of
activity, which grow and merge. This can be understood by noting that
the initial well-coordinates are distributed within an interval of
width larger than 1. The evolution towards positive threshold via the
ZFA, while conserving the fractional part of the well-coordinates,
gradually reduces this width by successively pruning the integer parts
of the well coordinates.  This means that while avalanches terminate
at sites with low well-coordinates, these values are often so low,
that their increments by $+1$ at the avalanche termination, as
prescribed by Proposition \ref{prop:avalanchewaves}{\em (ii)}, will
not make them avalanche initiation sites for the {\em next}
avalanche. Rather, there will be other sites with higher $z$-values
that serve as avalanche initiators.  This situation will continue
until such sites have been depleted sufficiently that the termination
sites of the previous avalanches do initiate the next avalanche.  This
is the major difference from the threshold-to-threshold evolution
where---due to the nature of the initial configuration---this
termination/initiation pattern is observed immediately from the start.
It was this observation that allowed for a description in terms of a
record-breaking process.  The behavior of the evolution starting form
a flat initial configuration is more interesting, but also more
difficult to describe precisely.

Another set of interesting questions relate to hysteretic behavior as
the force is raised and lowered, a feature previously observed in CDW
simulations \cite{Middleton93}.  For his recent master's thesis, Terzi
\cite{Terzi13} studied numerically hysteresis in the toy model. In the
toy model this occurs when the external force goes through a sequence
of force increments and decrements after which it returns to its
initial value. In terms of the ZFA evolution this amounts to running
this algorithm in the backward direction: Algorithm \ref{alg:ZFA} with
obvious modifications corresponding to force decrements. Starting from
a $(\pm)$-threshold configuration and applying sequences of forward
and backwards steps of the ZFA, Terzi finds that the total number of
reachable configurations scales like $L^{3/2}$.  One might hope that
for the toy model, this can be shown analytically, but this is not yet
done.  Terzi has additionally shown that the hysteretic behavior of
the toy model exhibits the return point memory effect. This is a
direct consequence of the no crossing property of the evolution
\cite{Sethna93}, which for our model is guaranteed by Lemma
\ref{lem:ZFAnoncrossing}.

The approach to the depinning transition using renormalization group
ideas \cite{NarayanFisher92,LeDoussal02,Ertas} suggests universality
of the behavior, near the transition. In particular it is believed
that the values of the scaling exponents should depend little on the
microscopic details of the underlying model.
The toy model serves as a good example to test these assumptions.
Here we find that depending on the initial configuration chosen, the
evolution to threshold and the corresponding scaling behavior is
markedly different. While in the threshold-to-threshold evolution the
correlation length near threshold diverges as $\xi \sim X^{-1}$ and
the quantity that exhibits scaling is the cumulative avalanche size
$\Sigma$ which scales as $\Sigma \sim X^{-2}$, in the evolution from a
flat initial condition to threshold we find numerical results
consistent with $\xi \sim X^{-2}$ and $P \sim X^{-3}$, which moreover
agrees with the renormalization group based prediction of Narayan {\em
  et al.}, \cite{NarayanFisher92,NarayanMiddleton94}.  The toy model
illustrates that the choice of initial condition can result in
dramatically different dynamics, leading to these disparate exponents.

Another type of universality is the robustness of our results when we
change the law of the underlying disorder. Theorem \ref{thm:toythresh}
can address immediately any randomness which is mutually absolutely
continuous with the $\vec{\alpha}$ considered here and thus one can
ask whether scaling still holds, and if so, how the scaling exponents
characterizing the correlation length and the cumulative avalanche
sizes, depend on the probability laws for the underlying disorder.

Generalization to higher dimensions is a more serious undertaking, as
the traditional two-dimensional sandpile is already much more
intricate than its one-dimensional relative \cite{Redig05}.  On the
other hand, the randomness could conceivably be helpful: the size of
the set of recurrent sandpile configurations should be smaller if the
sites on the lattice are no longer identical.  The authors hope to
consider this matter, and others mentioned above, in future work.

\section*{Acknowledgements}

DK and MM thank F.\ Rezakhanlou for stimulating discussions. DK
acknowledges the hospitality of the Istanbul Center for Mathematical
Sciences (IMBM) and the Mathematics and Physics Departments of Bo{\u
  g}azi{\c c}i University. MM acknowledges discussions with H.J.\
Jensen, M.M.\ Terzi, P.B. Littlewood, S.N. Coppersmith and A.\
Bovier. He thanks the Berkeley Math department for their kind
hospitality during his sabbatical stay.  This work was supported in
part by by NSF grant DMS-1106526 and Bo{\u g}azi{\c c}i University
grant 12B03P4.

\bibliographystyle{unsrt}%
\bibliography{km_long}%

\end{document}